\documentclass{lmcs}
\pdfoutput=1

\usepackage{lastpage}
\lmcsdoi{20}{1}{16}
\lmcsheading{}{\pageref{LastPage}}{}{}%
{Dec.~08,~2022}{Feb.~21,~2024}{}

\usepackage[utf8]{inputenc}


\usepackage{hyperref}
\usepackage[capitalize]{cleveref}
\usepackage{stmaryrd}
\usepackage{booktabs}
\usepackage{xcolor}
\usepackage{amssymb}
\usepackage{tikz}
\usetikzlibrary{arrows.meta}

\renewcommand{\models}{\vDash}
\newcommand{\nmodels}{\nvDash}

\newcommand{\prog}{\mathcal{P}}

\newcommand{\tbox}{\mathcal{T}}
\newcommand{\Sch}{\mathcal{S}}
\renewcommand{\circ}/
\newcommand{\comp}/
\newcommand\concept{\phi}

\newcommand\voc{\Sigma}

\newcommand{\defc}[1]{\llbracket #1 \rrbracket}
\newcommand{\iexpr}[2]{\llbracket #1 \rrbracket^{#2}}      
\newcommand{\semi}[1]{\iexpr{#1}I}
\newcommand{\semj}[1]{\iexpr{#1}J}
\newcommand{\sem}[1]{\iexpr{#1}G}

\newcommand{\sems}[1]{\iexpr{#1}{\Istar}}
 
\newcommand{\Istar}{I}
\newcommand{\const}[1]{\{#1\}}

\newcommand{\id}{\mathit{id}}

\newcommand{\eq}{\mathit{eq}}
\newcommand{\disj}{\mathit{disj}}
\newcommand{\closed}{\mathit{closed}}
\newcommand{\lang}{\mathcal{L}}
\newcommand\dom{\Delta}
\newcommand{\geqn}[3]{\mathop\geq\nolimits_{#1}#2.#3}
\newcommand{\leqn}[3]{\leq_{#1}#2.#3}

\newcommand{\fulleq}{\mathit{full\text{-}eq}}
\newcommand{\fulldisj}{\mathit{full\text{-}disj}}

\newcommand{\prev}{\mathit{prev}}
\newcommand{\nxt}{\mathit{next}}

\begin{document}
\title[Expressiveness of SHACL Features]{Expressiveness of SHACL
  Features and Extensions for Full Equality and Disjointness
  Tests}\thanks{This paper is an extended version of our work
  presented at ICDT 2022 \cite{icdt/BogaertsJV22}.}
\author[B.~Bogaerts]{Bart Bogaerts\lmcsorcid{0000-0003-3460-4251}}[a]
\author[M.~Jakubowski]{Maxime
  Jakubowski\lmcsorcid{0000-0002-7420-1337}}[a,b] \author[J.~Van den
Bussche]{Jan {Van den Bussche}\lmcsorcid{0000-0003-0072-3252}}[b]

\address{Vrije Universiteit Brussel, Belgium}	
\email{bart.bogaerts@vub.be}  

\address{Universiteit Hasselt, Belgium}	
\email{maxime.jakubowski@uhasselt.be}
\email{jan.vandenbussche@uhasselt.be} 

\begin{abstract}
  SHACL is a W3C-proposed schema language for expressing structural
  constraints on RDF graphs.  Recent work on formalizing this language
  has revealed a striking relationship to description logics.  SHACL
  expressions can use three fundamental features that are not so
  common in description logics.  These features are equality tests;
  disjointness tests; and closure constraints.  Moreover, SHACL is
  peculiar in allowing only a restricted form of expressions
  (so-called targets) on the left-hand side of inclusion constraints.

  The goal of this paper is to obtain a clear picture of the impact
  and expressiveness of these features and restrictions.  We show that
  each of the three features is primitive: using the feature, one can
  express boolean queries that are not expressible without using the
  feature.  We also show that the restriction that SHACL imposes on
  allowed targets is inessential, as long as closure constraints are
  not used.

  In addition, we show that enriching SHACL with ``full'' versions of
  equality tests, or disjointness tests, results in a strictly more
  powerful language.
\end{abstract}

\maketitle

\section{Introduction}

On the Web, the Resource Description Framework (RDF
\cite{rdf11primer}) is a standard format for representing
knowledge and publishing data.  RDF represents information in the
form of directed graphs, where labeled edges indicate properties
of nodes.  To facilitate more effective access and exchange, it
is important for a consumer of an RDF graph to know what
properties to expect, or, more generally, to be able to rely on
certain structural constraints that the graph is guaranteed to
satisfy.  We therefore need a declarative language in which such
constraints can be expressed formally.  In database terms, we
need a schema language.

Two prominent proposals in this vein have been ShEx \cite{shex}
and SHACL \cite{shacl}.  Both approaches use formulas that
express the presence or absence of certain properties of a
node or its neighbors in the graph.  Such formulas are called
``shapes.'' When we evaluate a shape on a node, that node is
called the ``focus node.''  Some
examples of shapes, expressed for now in English, could be the
following:\footnote{In real RDF, names of properties and nodes
must conform to IRI syntax, but in this paper, to avoid
clutter, we take the liberty to
use simple names.}
\begin{enumerate}
  \item ``The focus node has a \textsf{phone} property, but no
    \textsf{email}.''
  \item ``The focus node has at least five incoming
    \textsf{managed-by} edges.''
  \item ``Through a path of \textsf{friend} edges, the focus node
  can reach a node with a \textsf{CEO-of} edge to the node
  \textsf{Apple}.''
  \item ``The focus node has at least one \textsf{colleague} who is
    also a \textsf{friend}.''
  \item ``The focus node has no other properties than
    \textsf{name}, \textsf{address}, or \textsf{birthdate}.''
\end{enumerate}

In this paper, we look deeper into SHACL, the language
recommended by the World Wide Web Consortium.
We do not use the actual SHACL syntax, but work
with the elegant formalization proposed by Corman, Reutter and
Savkovic \cite{corman}, and used in subsequent works by several
authors \cite{andresel,leinberger,shaclsatsouth}.  That
formalization reveals a striking similarity between shapes on the
one hand, and concepts, familiar from description logics
\cite{dlintro}, on the other hand.  The similarity between SHACL
and description logics runs even deeper when we account for
targeting, which is the actual mechanism to express constraints
on an RDF graph using shapes.

Specifically, a non-recursive \emph{shape schema\footnote{Real
SHACL uses the term \emph{shapes graph} instead of shape schema.}}
is essentially a finite list
of shapes, where each shape $\phi$ is additionally associated
with a target query $q$.  An RDF graph $G$ is said to conform to
such a schema if for every target--shape combination $(q,\phi)$,
and every node $v$ returned by $q$ on $G$, we have that $v$
satisfies $\phi$ in $G$.  Let us see some examples of
target--shape pairs, still expressed in English:
\begin{enumerate}
\setcounter{enumi}{5}
\item
``Every node of type \textsf{Person} has an \textsf{email} or
\textsf{phone} property.''  Here, the target query returns all nodes
with an edge labeled \textsf{type} to node \texttt{Person}; the shape
    checks that the focus node
    has an \textsf{email} or \textsf{phone} property.
\item
  ``Different nodes never have the same \textsf{email}.'' Here
    the target query returns all nodes with an incoming
    \textsf{email} edge, and the shape checks that the focus node
    does not have two or more incoming \textsf{email} edges.
\item
  ``Every mathematician has a finite Erd\H os number.''
Here the target query returns all nodes of type
  \textsf{Mathematician}, and
    the shape checks that the focus node can reach the
    node \textsf{Erd\H os} by a path that matches the regular
    expression $(\textsf{author}^-/\textsf{author})^*$.  Here,
    the minus superscript denotes an inverse edge.
\end{enumerate}

Interestingly, and apparent in the examples 6--8, the
target queries considered for this purpose in SHACL, as well as
in ShEx, actually correspond to simple cases of shapes.  It is
then only a small step to consider \emph{generalized} shape schemas as
finite sets of inclusion statements of the form $\phi_1 \subseteq
\phi_2$, where $\phi_1$ and $\phi_2$ are shapes.  Since, as noted
above, shapes correspond to concepts, we thus see that shape
schemas correspond to TBoxes in description logics.

We stress that the task we are focusing on in this paper is
checking conformance of RDF graphs against shape schemas.  Every
shape schema $\mathcal{S}$ defines a decision problem: given an RDF
graph $G$, check whether $G$ conforms to $\mathcal{S}$.  In database
terms, we are processing a \emph{boolean query} on a graph
database.  In description logic terms, this amounts to
\emph{model checking} of a TBox: given an interpretation, check
whether it satisfies the TBox.  Thus our focus is a bit different
from that of typical applications of description logics.  There,
facts are declared in ABoxes, which should not be confused with
interpretations.  The focus is then on higher reasoning tasks,
such as checking satisfiability of an ABox and a TBox, or
deciding logical entailment.

Given the above context, let us now look in more detail at the
logical constructs that can be used to build shapes.  Some of
these constructs are well known concept constructors
from expressive description logics \cite{expressivedl}:
the boolean connectives; constants;
qualified number restriction (a combination of existential
quantification and counting); and regular path expressions with inverse.  
To illustrate, example shapes 1--3 are expressible as follows:
\begin{enumerate}
\item
$\geqn 1{\textsf{phone}}\top \land
\neg \geqn 1{\textsf{email}}\top$. This uses qualified number
restriction with count 1 (so essentially existential
quantification), conjunction, and negation; $\top$ stands for true.
\item
$\geqn 5{\textsf{managed-by}^-}\top$. This uses counting to 5,
and inverse.
\item
$\geqn 1{\textsf{friend}^*/\textsf{CEO-of}}{\const{\textsf{Apple}}}$.
This uses a regular path expression and the constant
\textsf{Apple}.
\end{enumerate}

However, SHACL also has three specific logical features that are less
common:
\begin{description}
  \item[Equality] The shape $\eq(E,r)$, for a
    path expression $E$ and a property $r$, tests equality of
    the sets of nodes reachable from the focus node by an
    $r$-edge on the one hand, and by an $E$-path on the other
    hand.  
  \item[Disjointness] A similar shape $\disj(E,r)$ tests disjointness
    of the two sets of reachable nodes.  To illustrate, example
    shape~4 is expressed as
    $ \neg \disj(\textsf{colleague},\textsf{friend}) $.
  \item[Closure constraints] RDF graphs to be checked for
    conformance against some shape schema need not obey some fixed
    vocabulary.  Thus SHACL provides shapes of the form
    $\closed(R)$, with $R$ a finite set of
    properties, expressing that the focus node has no properties
    other than those in $R$.  This was already illustrated as
    example shape~5, with $R=\{
    \textsf{name}, \textsf{address}, \textsf{birthdate}\}$.
\end{description}

Our goal in this paper is to clarify the impact of the above three
features on the expressiveness of SHACL as a language for boolean
queries on graph databases.  Thereto, we offer the following
contributions.
\begin{itemize}
\item We show that each of the three features is primitive in a strong
  sense.  Specifically, for each feature, we exhibit a boolean query
  $Q$ such that $Q$ is expressible by a single target--shape pair,
  using only the feature and the basic constructs; however, $Q$ is not
  expressible by any generalized shape schema when the feature is
  disallowed.
\item We also clarify the significance of the restriction that SHACL
  puts on allowed targets.  We observe that as long as closure
  constraints are not used, the restriction is actually insignificant.
  Any generalized shape schema, allowing arbitrary but closure-free
  shapes on the left-hand sides of the inclusion statements, can be
  equivalently written as a shape schema with only targets on the
  left-hand sides.  However, allowing closure constraints on the
  left-hand side strictly adds expressive power.
\item We additionally show that ``full'' variants of equality tests or
  disjointness tests result in strictly more expressive
  languages. This result anticipates planned extensions of SHACL
  \cite{dash, shacl_extensions}.
\item Our results continue to hold when the definition of
  \emph{recursive} shapes is allowed, provided that recursion through
  negation is stratified.
\end{itemize}

This paper is organized as follows.  Section~\ref{secdefs} presents
clean formal definitions of non-recursive shape schemas, building on
and inspired by the work of Andre\c sel, Corman, et al.\ cited above.
Section~\ref{secexpr} and Section~\ref{sec:targetres} present our
results, and Section \ref{sec:fulltests} extends our result for
``full'' equality and disjointness tests.  Section~\ref{secrec}
presents the extension to stratified recursion.  Section~\ref{seconc}
offers concluding remarks.

\section{Shape schemas} \label{secdefs}

In this section we define shapes, RDF graphs, shape schemas, and the
conformance of RDF graphs to shape schemas.  Perhaps curiously to
those familiar with SHACL, our treatment for now omits \emph{shape
  names}.  Shape names are redundant as far as expressive power is
concerned, as long as we are in a non-recursive setting, because shape
name definitions can then always be unfolded.  Indeed, for clarity of
exposition, we have chosen to work first with non-recursive shape
schemas.  Section~\ref{secrec} then presents the extension to
recursion (and introduces shape names in the process).  We point out
that the W3C SHACL recommendation only considers non-recursive shape
schemas.

\paragraph*{Node and property names}
From the outset we assume two disjoint, infinite universes $N$ and $P$
of \emph{node names} and \emph{property names},
respectively.\footnote{In practice, node names and property names are
  IRIs \cite{rdf11primer}, so the disjointness assumption would not
  hold.  However, this assumption is only made for simplicity of
  notation; it can be avoided if we make our notation for vocabularies
  and interpretations (see below) more complicated.}

\subsection{Shapes}

We define \emph{path expressions} $E$ and \emph{shapes} $\phi$ 
by the following grammar:
\begin{align*}
  E & ::=  \id \mid p \mid p^- \mid E \cup E \mid E \circ E \mid E^* \\
  \concept &::= \top \mid \const c \mid \concept \land \concept \mid
  \concept \lor \concept \mid \neg \concept \mid
  {}\geqn{n}{E}{\concept} \mid \eq(E,p)\mid \disj(E,p) \mid
\closed(R)
\end{align*}
Here, $p$ and $c$ stand for property names and node names,
respectively; $n$ stands for nonzero natural numbers; and $R$ stands
for finite sets of property names.  A node name $c$ is also referred
to as a \emph{constant}. In $p^-$, $-$ stands for inverse.

\begin{description}
\item[Abbreviation] We will use $\exists E.\phi$ as an abbreviation
  for $\geqn 1E\phi$.
\end{description}

\begin{rem}
  Real SHACL supports some further shapes which have to do with tests
  on IRI constants and literals, as well as comparisons on numerical
  values and language tags.  Like other work on the formal aspects of
  SHACL, we abstract these away, as many questions are already
  interesting without these features.
\end{rem}

\begin{rem}
  Universal quantification $\forall E.\phi$ could be introduced as an
  abbreviation for $\neg \exists E.\neg \phi$.  Likewise,
  $\leqn nE\phi$ may be used as an abbreviation for
  $\neg \geqn {n+1}E\phi$.
\end{rem}

\begin{rem} \label{rem:id}
  In our formalization, a path expression can be `$\id$'. We show in
  Lemma~\ref{lem:id} that every path expression is equivalent to
  $\id$, $E'\cup\id$ or $E'$, where $E'$ does not use $\id$. In real
  SHACL, it is possible to write $E'\cup\id$ using ``zero-or-one''
  path expressions. Explicitly writing $\id$ is not possible, but this
  poses no problem. Path expressions can only appear in counting
  quantifiers, equality and disjointness shapes. The shape
  $\geqn{n}{\id}{\phi}$ is clearly equivalent to $\phi$ if $n=1$,
  otherwise, it is equivalent to $\neg\top$. The shapes $\eq(E,p)$ or
  $\disj(E,p)$ where $E$ is $\id$ are implicitly expressible in SHACL
  by writing the equality or disjointness constraint in node shapes,
  rather than property shapes.
\end{rem}

A \emph{vocabulary} $\Sigma$ is a subset of $N \cup P$.  A path
expression is said to be \emph{over} $\Sigma$ if it only uses property
names from $\Sigma$.  Likewise, a shape is over $\Sigma$ if it only
uses constants from $\Sigma$ and path expressions over $\Sigma$.

Following common practice in logic, shapes are evaluated in
\emph{interpretations}.  We recall the familiar definition of an
interpretation.  Let $\Sigma$ be a vocabulary.  An interpretation $I$
over $\Sigma$ consists of
\begin{itemize}
\item a set $\dom^I$, called the \emph{domain} of $I$;
\item for each constant $c \in \Sigma$, an element
  $\semi c \in \dom^I$; and
\item for each property name $p \in \Sigma$, a binary relation
  $\semi p$ on $\dom^I$.
\end{itemize}

The semantics of shapes is now defined as follows.
\begin{itemize}
\item On any interpretation $I$ as above, every path expression $E$
  over $\Sigma$ evaluates to a binary relation $\semi E$ on $\dom^I$,
  defined in Table~\ref{tab:sempath}.
\item Now for any shape $\phi$ over $\Sigma$ and any element
  $a \in \dom^I$, we define when \emph{$a$ conforms to $\phi$ in $I$},
  denoted by $I,a \models \phi$.  For the boolean operators $\top$
  (true), $\land$ (conjunction), $\lor$ (disjunction), $\neg$
  (negation), the definition is obvious.  For the other constructs,
  the definition is given in Table~\ref{tab:sem}, taking note of the
  following:

\begin{itemize}
\item We use the notation $R(x)$, for a binary relation $R$, to denote
  the set $\{y \mid (x,y) \in R\}$. We apply this notation to the case
  where $R$ is of the form $\semi E$.
\item We also use the notation $\sharp X$ for the cardinality of a set
  $X$.
\end{itemize}
\item For a shape $\phi$ and interpretation $I$, the notation
\[ \semi\phi := \{a \in \dom^I \mid I,a \models \phi \} \] will be
convenient.
\end{itemize}

\begin{table}
  \centering
  \begin{tabular}{ll}
    \toprule
    $E$ & $\iexpr{E}{I}$ \\
    \midrule
    $\id$ & $\{(x,x) \mid x \in \dom^I\}$ \\
    $p^-$ & $\{(a,b) \mid (b,a) \in \iexpr{p}{I}\}$\\
    ${E_1 \cup E_2}$ & $\iexpr{E_1}{I} \cup \iexpr{E_2}{I}$\\
    ${E_1 \circ E_2}$ & $\{(a,b) \mid \exists x: (a,x) \in \iexpr{E_1}{I}\land (x,b) \in \iexpr{E_2}{I} \}$ \\
    ${E^*}$ & the reflexive-transitive closure of $\iexpr{E}{I}$\\
    \bottomrule
  \end{tabular}
  \caption{Semantics of path expressions.} \label{tab:sempath}
\end{table}

\begin{table}
\centering
\begin{tabular}{ll}
\toprule
$\phi$ & $I,a \models \phi$ if: \\
\midrule
$\const c$ & $a=\semi c$ \\
$\geqn{n}{E}{\psi}$ &
$\sharp \{b \in \semi E(a) \mid I,b \models \psi\} \geq n$ \\
$\eq(E,p)$ &
the sets $\semi E(a)$ and $\semi p(a)$ are equal \\
$\disj(E,p)$ &
the sets $\semi E(a)$ and $\semi p(a)$ are disjoint \\
$\closed(R)$ & $\semi p(a)$ is empty for each $p \in \Sigma-R$ \\
\bottomrule
\end{tabular}
\caption{Conditions for conformance of a node to a shape.}
\label{tab:sem}
\label{tabdef}
\end{table}

\begin{exa}
  \label{ex:running1} In the Introduction we already gave three
  examples (1), (2), and (3) of shapes expressed in English and the
  formal syntax. The target query of example~(7) from the Introduction
  can be expressed as the shape $\exists \textsf{email}^-.\top$. The
  shape of example~(7) can be written as
  $\leqn 1{\textsf{email}^-}\top$.
\end{exa}

\subsection{Graphs and their interpretation}

Remember from Table~\ref{tab:sem} that a shape $\closed(R)$ states
that the focus node may only have outgoing properties that are
mentioned in $R$. It may appear that such a shape is simply
expressible as the conjunction of $\neg \exists p.\top$ for
$p \in \Sigma-R$.  However, since shapes must be finite formulas, this
only works if $\Sigma$ is finite.  In practice, shapes can be
evaluated over arbitrary RDF graphs, which can involve arbitrary
property names (and node names), not limited to a finite vocabulary
that is fixed in advance.

Formally, we define a \emph{graph} as a finite set of triples of the
form $(a,p,b)$, where $p$ is a property name and $a$ and $b$ are (not
necessarily distinct) node names.  We refer to the node names
appearing in a graph $G$ simply as the \emph{nodes} of $G$; the set of
nodes of $G$ is denoted by $N_G$.  A pair $(a,b)$ with $(a,p,b) \in G$
is referred to as an \emph{edge}, or a \emph{$p$-edge}, in $G$.

We now canonically view any graph $G$ as an interpretation over the
\emph{full} vocabulary $N \cup P$ as follows:
\begin{itemize}
\item
$\dom^G$ equals $N$ (the universe of all node names).
\item
$\sem c$ equals $c$ itself, for every node name $c$.
\item
$\sem p$ equals the set of $p$-edges in $G$, for every property name $p$.
\end{itemize}
Note that since graphs are finite, $\sem p$ will be empty for all but
a finite number of $p$'s.

Given this canonical interpretation, path expressions and shapes
obtain a semantics on all graphs $G$.  Thus for any path expression
$E$, the binary relation $\sem E$ on $N$ is well-defined; for any
shape $\phi$ and $a \in N$, it is well-defined whether or not
$G,a \models \phi$; and we can also use the notation $\sem \phi$.

\begin{rem}
  Since a graph is considered to be an interpretation with the
  infinite domain $N$, it may not be immediately clear that shapes can
  be effectively evaluated over graphs.  Adapting well-known methods,
  however, we can reduce to a finite domain over a finite vocabulary
  \cite[Theorem 5.6.1]{ahv_book}, \cite{4rus,hs_domind}.  Formally,
  let $\phi$ be a shape and let $G$ be a graph.  Recall that $N_G$
  denotes the set of nodes of $G$; similarly, let $P_G$ be the set of
  property names appearing in $G$.  Let $C$ be the set of constants
  mentioned in $\phi$.  We can then form the finite vocabulary
  $\Sigma = N_G \cup C \cup P_G$.  Now define the interpretation
  $\Istar$ over $\Sigma$ as follows:
  \begin{itemize}
  \item $\dom^{\Istar} = N_G \cup C \cup \{\star\}$, where $\star$ is an
    element not in $N$;
  \item $\sems c = c$ for each node name $c \in \Sigma$;
  \item $\sems p = \sem p$ for each property name $p \in \Sigma$.
  \end{itemize}
  Note that no constant symbol names $\star$ in $\Istar$.  Then for
  every $x \in N_G \cup C$, one can show that $x \in \sem \phi$ if and
  only if $x \in \sems \phi$.  For all other node names $x$, one can
  show that $x \in \sem \phi$ if and only if $\star \in \sems \phi$.
\end{rem}

\newcommand{\ex}{\mathit{ex}}

\begin{exa}[Example \ref{ex:running1} continued] \label{ex:running2}
  Consider the graph $G_\ex$ depicted in
  Figure~\ref{fig:exgraph}. This graph can be seen as the
  interpretation $I_\ex$ with an infinite domain containing the
  elements $a$, $b$, $m_1$, and $m_2$. It interprets the predicate
  name $\textsf{email}$ as $\{(a,m_1), (b,m_1), (b,m_2)\}$ and all
  other predicate names as the empty set.  If we look at the
  interpretation of example~(7) from the Introduction in $I_\ex$, we
  have $\iexpr{\leqn{1}{\textsf{email}^-}\top}{I_\ex}=\{m_1\}$ for the
  shape, and
  $\iexpr{\exists \textsf{email}^-.\top}{I_\ex}=\{m_1,m_2\}$ for the
  target.
  \end{exa}

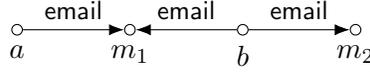
\begin{figure}[t]
\centering
\begin{tikzpicture}
\tikzstyle{grnode}=[draw,circle,fill=white,minimum size=4pt,
                            inner sep=0pt]
\draw (0,0) node [grnode] (a) [label=below:$a$] {};
\draw (1.5,0) node [grnode] (m1) [label=below:$m_1$] {};
\draw (3,0) node [grnode] (b) [label=below:$b$] {};
\draw (4.5,0) node [grnode] (m2) [label=below:$m_2$] {};
  
\draw [-{Latex[length=2mm]}] (a) -- (m1) node [sloped,midway,above ] {\small{$\textsf{email}$}};
\draw [-{Latex[length=2mm]}] (b) -- (m2) node [sloped,midway,above ] {\small{$\textsf{email}$}};
\draw [-{Latex[length=2mm]}] (b) -- (m1) node [sloped,midway,above ] {\small{$\textsf{email}$}};

\end{tikzpicture}
\caption{An example graph $G_\ex$}\label{fig:exgraph}
\end{figure}

\subsection{Targets and shape schemas}

SHACL identifies four special forms of shapes and calls them
\emph{targets}:  
\begin{description}
\item[Node targets] $\const c$ for any constant $c$.
\item[Class-based targets] $\exists
\textsf{type}/\textsf{subclass}^*.\const c$ for any constant $c$.
Here, $\textsf{type}$ and $\textsf{subclass}$ represent distinguished
IRIs from the RDF Schema vocabulary \cite{rdf11primer}.
\item[Subjects-of targets] $\exists p.\top$ for any property name
$p$.
\item[Objects-of targets] $\exists p^-.\top$ for any property name
$p$.
\end{description}

We now define a \emph{generalized shape schema} (or shape schema
for short) as a finite set of inclusion statements, where an
inclusion statement is of the form $\phi_1 \subseteq \phi_2$,
with $\phi_1$ and $\phi_2$ shapes.  A \emph{target-based shape
schema} is a shape schema that only uses targets, as defined
above, on the left-hand sides of its inclusion statements.  This
restriction corresponds to the shape schemas considered in real
SHACL.

As already explained in the Introduction, a graph $G$ \emph{conforms}
to a shape schema $\Sch$, denoted by $G \models \Sch$, if
$\sem{\phi_1}$ is a subset of $\sem{\phi_2}$, for every statement
$\phi_1 \subseteq \phi_2$ in $\Sch$.

Thus, any shape schema $\Sch$ defines the class of graphs that conform
to it.  We denote this class of graphs by
\[ \defc{\Sch} := \{\text{graph $G$} \mid G \models \Sch\}. \]
Accordingly, two shape schemas $\Sch_1$ and $\Sch_2$ are said to be
\emph{equivalent} if $\defc{\Sch_1} = \defc{\Sch_2}$.

\begin{exa}[Example \ref{ex:running2} continued] \label{ex:running3}
  Constraint (7) from the Introduction can be written as:
  \[
    \exists \textsf{email}^-.\top \subseteq {}\leqn{1}{\textsf{email}}\top
  \]
  We can see that $G_{ex}$ does not satisfy the constraint, because
  $\{m_1,m_2\}\not\subseteq\{m_2\}$. However, if we remove the
  triple $(b,\textsf{email},m_1)$ from $G_{ex}$, then the shape is
  interpreted as $\iexpr{\leqn{1}{\textsf{email}}\top}{I_\ex} = \{m_1,m_2\}$
  and the constraint does hold, as
  $\{m_1,m_2\}\subseteq\{m_1, m_2\}$.
\end{exa}

\section{Expressiveness of SHACL features} \label{secexpr}

When a complicated but influential new tool is proposed in the
community, in our case SHACL, we feel it is important to have a solid
understanding of its design.  Concretely, as motivated in the
Introduction, our goal is to obtain a clear picture of the relative
expressiveness of the features $\eq$, $\disj$, and $\closed$.  Our
methodology is as follows.

A \emph{feature set} $F$ is a subset of $\{\eq,\disj,\closed\}$.  The
set of all shape schemas using only features from $F$, besides the
standard constructs, is denoted by $\lang(F)$.  In particular, shape
schemas in $\lang(\emptyset)$ use only the standard constructs and
none of the three features.  Specifically, they only involve shapes
built from boolean connectives, constants, and qualified number
restrictions, with path expressions built from property names, $\id$
and the standard operators union, composition, and Kleene
star. 

We say that feature set $F_1$ is \emph{subsumed} by feature set $F_2$,
denoted by $F_1 \preceq F_2$, if every shape schema in $\lang(F_1)$ is
equivalent to some shape schema in $\lang(F_2)$.  As it will turn out,
\begin{equation}
  \label{eq:1}
  \tag{$*$}
  F_1 \preceq F_2 \quad \Leftrightarrow \quad F_1 \subseteq F_2,
\end{equation}
or intuitively, ``every feature counts.''
Note that the implication from right to left is trivial, but the
other direction is by no means clear from the outset.

More specifically, for every feature, we introduce a class of graphs,
as follows.  In what follows we fix some property name $r$.
\begin{description}
\item[Equality] $Q_{\eq}$ is the class of graphs where all $r$-edges
  are symmetric.  Note that $Q_{\eq}$ is definable in $\lang(\eq)$ by
  the single, target-based, inclusion statement
  $\exists r.\top \subseteq \eq(r^-,r)$.
\item[Disjointness] $Q_{\disj}$ is the class of graphs where all
  nodes with an outgoing $r$-edge have at least one symmetric
  $r$-edge.  This time, $Q_{\disj}$ is definable in $\lang(\disj)$, by
  the single, target-based, inclusion statement
  $\exists r.\top \subseteq \neg \disj(r^-,r)$.
\item[Closure] $Q_{\closed}$ is the class of graphs where for all
  nodes with an outgoing $r$-edge, all outgoing edges have label $r$.
  Again $Q_{\closed}$ is definable in $\lang(\closed)$ by the single,
  target-based, inclusion statement
  $\exists r.\top \subseteq \closed(r)$.
\end{description}

We establish the following theorem, from which the above equivalence
(\ref{eq:1}) immediately follows:

\begin{thm} 
  \label{bool}
  Let $X \in\{\eq,\disj,\closed\}$ and let $F$ be a feature set with
  $X \notin F$.  Then $Q_X$ is not definable in $\lang(F)$.
\end{thm}

For $X = \closed$, Theorem~\ref{bool} is proven differently than for
the other two features. First, we deal with the remaining features
through the following concrete result, illustrated in
Figure~\ref{figraphs}. The formal definition of the graphs illustrated
in Figure~\ref{figraphs} for $X=\disj$ will be provided in
Definition~\ref{def:graph_disj}.

\begin{prop} \label{deprop} Let $X = \disj$ or $\eq$.  Let $\Sigma$ be
  a finite vocabulary including $r$, and let $m$ be a nonzero natural
  number.  There exist two graphs $G$ and $G'$ with the following
  properties:
\begin{enumerate}
  \item $G'$ belongs to $Q_X$, but $G$ does not.
  \item For every shape $\phi$ over $\Sigma$ such that $\phi$ does not
    use $X$, and $\phi$ counts to at most $m$, we have
    \[ \sem \phi = \iexpr{\phi}{G'}. \]
\end{enumerate}
\end{prop}
Here, ``counting to at most $m$'' means that all quantifiers $\geq_n$
used in $\phi$ satisfy $n\leq m$. For $X=\eq$, this proposition is
reformulated as Proposition~\ref{prop:disj}, and for $X=\disj$, this
proposition is reformulated as Proposition~\ref{prop:eq}.

To see that Proposition~\ref{deprop} indeed establishes
Theorem~\ref{bool} for the three features under consideration, we use
the notion of \emph{validation shape} of a shape schema.  This shape
evaluates to the set of all nodes that violate the schema.  Thus, the
validation shape is an abstraction of the ``validation report'' in
SHACL \cite{shacl}: a graph conforms to a schema if and only if the
validation shape evaluates to the empty set.  The validation shape can
be formally constructed as the disjunction of
$\phi_1 \land \neg \phi_2$ for all statements
$\phi_1 \subseteq \phi_2$ in the schema.

Now consider a shape schema $\Sch$ not using feature $X$.  Let $m$ be
the maximum count used in shapes in $\Sch$, and let $\Sigma'$ be the
set of constants and property names mentioned in $\Sch$.  Now given
$\Sigma=\Sigma'\cup\{r\}$ and $m$, let $G$ and $G'$ be the two graphs
exhibited by the Proposition, and let $\phi$ be the validation shape
for $\Sch$.  Then $\phi$ will evaluate to the same result on $G$ and
$G'$. However, for $\Sch$ to define $Q_X$, validation would have to
return the empty set on $G'$ but a nonempty set on $G$.  We conclude
that $\Sch$ does not define $Q_X$.

We will prove Proposition~\ref{deprop} for $X=\disj$ in
Section~\ref{secdisjproof}, and $X=\eq$ in
Section~\ref{sec:equality}. We will show Theorem~\ref{bool} for
$X=\closed$ in Section~\ref{seclosure}. However, we first need to
establish some preliminaries on path expressions.

\begin{figure}
\centering
\begin{tabular}{ccc}
\toprule
$X$ & $G$ & $G'$ \\
\midrule
$\eq$ &
\raisebox{-0.5\height}{\scalebox{0.3}{\includegraphics{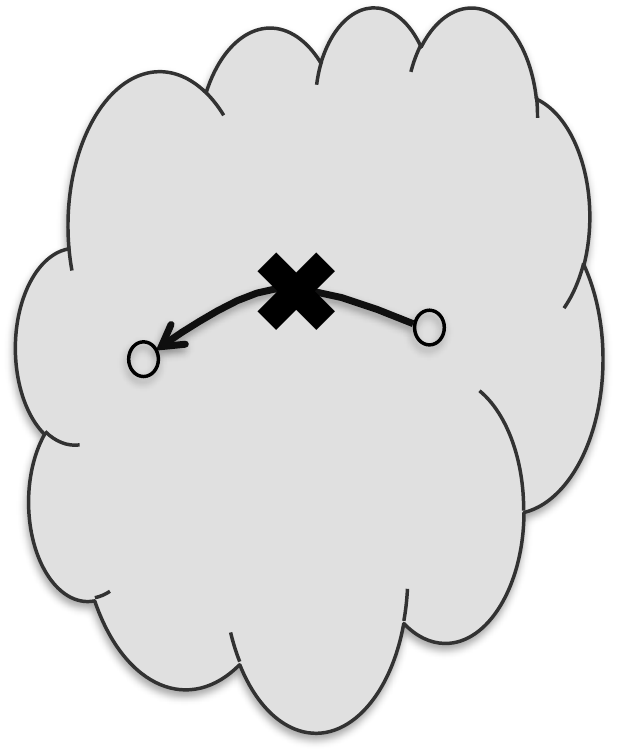}}}
&
\raisebox{-0.5\height}{\scalebox{0.3}{\includegraphics{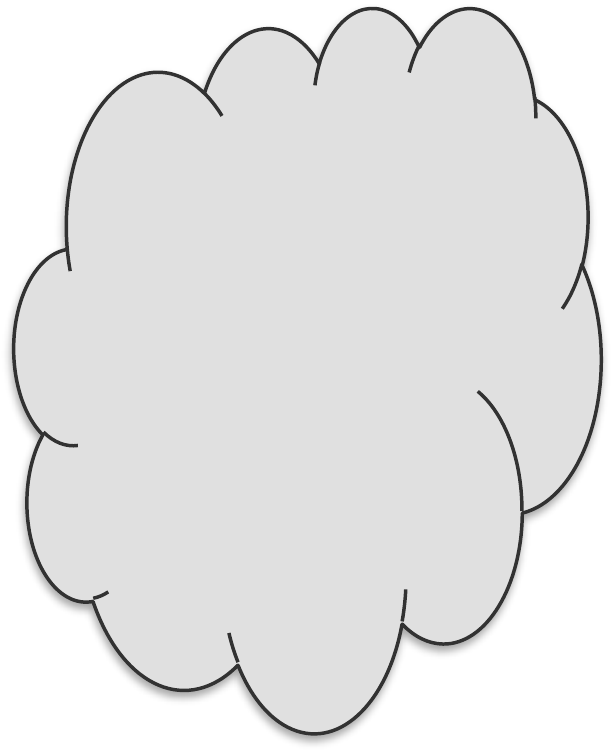}}}
\\[9ex]
$\disj$ &
\raisebox{-0.5\height}{\scalebox{0.2}{\includegraphics{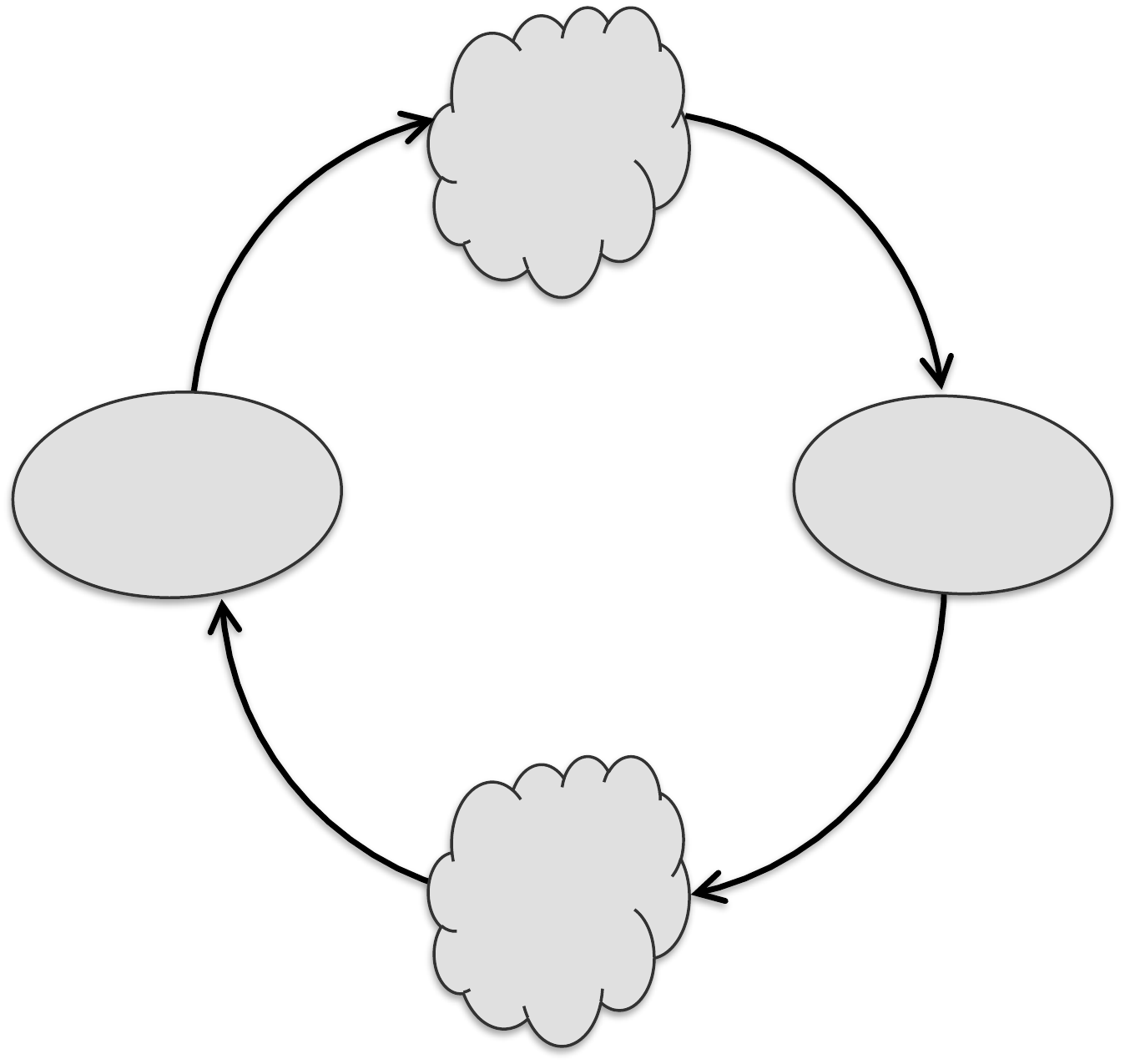}}}
&
\raisebox{-0.5\height}{\scalebox{0.2}{\includegraphics{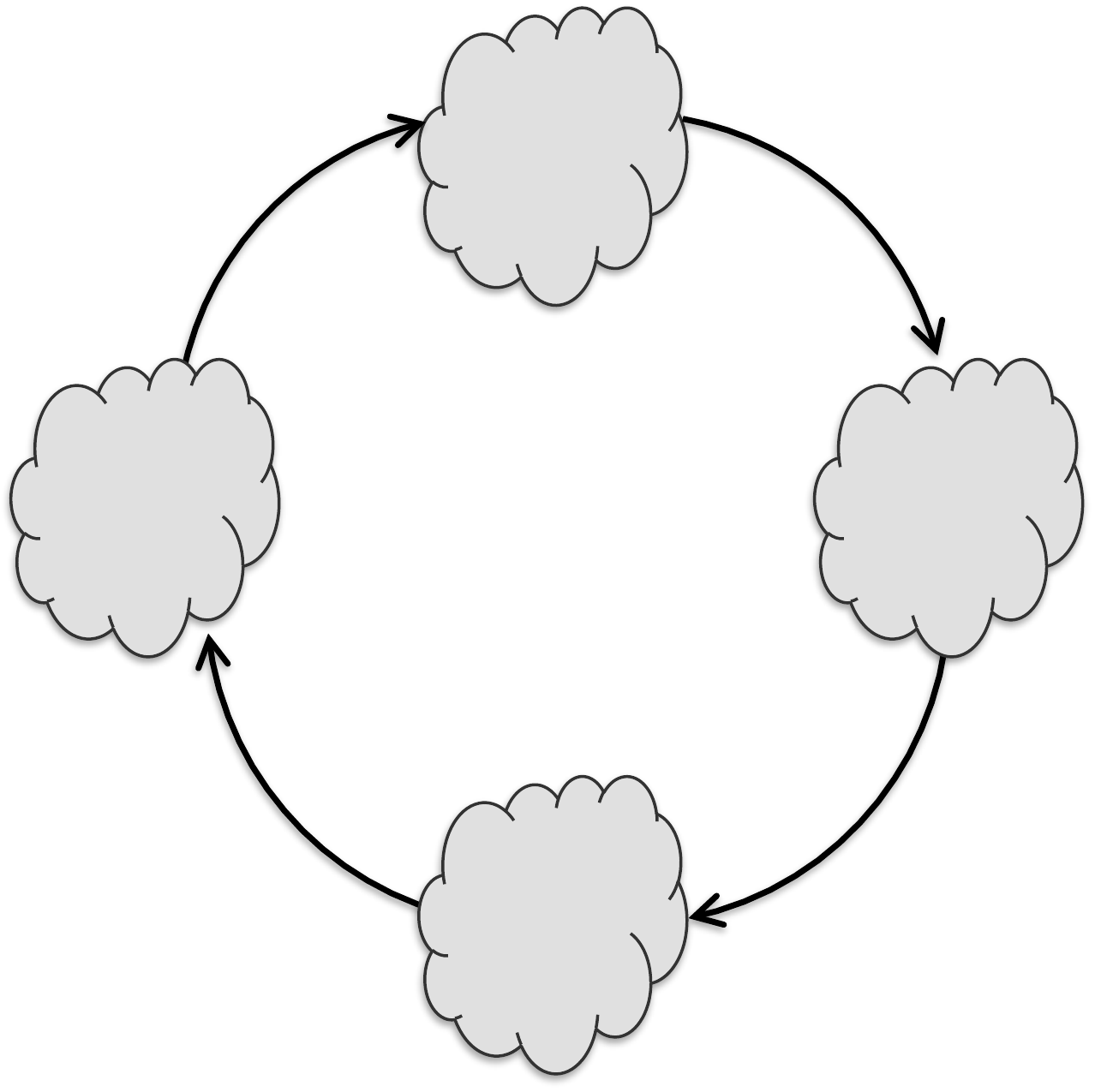}}}
\\
\bottomrule
\end{tabular}
\caption{Graphs used to prove
Proposition~\ref{deprop}.  The nodes are taken outside $\Sigma$.
For $X=\eq$, the cloud shown for $G'$ represents a complete
directed graph
on $m+1$ nodes, with self-loops,
and $G$ is the same graph with one directed edge
removed.  For $X=\disj$, in the picture for $G$, each cloud again
stands for a complete graph, but this time on $M=\max(m,3)$ nodes,
and without the self-loops. Each oval stands
for a set of $M$ separate nodes.  An arrow from one blob to the
next means that every node of the first blob has a directed edge
to every node of the next blob.
So, $G$ is a directed 4-cycle of alternating clouds and ovals, and
$G'$ is a directed 4-cycle of clouds.}
\label{figraphs}
\end{figure}

\subsection{Preliminaries on path expressions}
\label{sec:prelimpath}

We call a path expression $E$ \emph{equivalent} to a path expression
$E'$ when for every graph $G$, $\iexpr{E}{G}=\iexpr{E'}{G}$. We call a
path expression $E$ \emph{$id$-free} whenever $\id$ is not present in
the expression.

\begin{lem}
  \label{lem:id}
  Every path expression $E$ is equivalent to: $\id$, or $E'\cup\id$, or
  $E'$ where $E'$ is an $\id$-free path expression.
\end{lem}

\begin{proof}
  The proof is by induction on the structure of $E$. When $E$ is
  $id$-free or $\id$, the claim directly follows. We consider the
  following inductive cases:
  \begin{itemize}
  \item $E$ is $E_1\comp E_2$. By induction, we consider nine
    cases. When both $E_1$ and $E_2$ are $\id$-free, $E$ is
    $\id$-free. Whenever $E_1$ is $\id$, clearly $E$ is equivalent to
    $E_2$. Analogously, whenever $E_2$ is $\id$, $E$ is equivalent to
    $E_1$.

    Consider the two cases where $E_1$ is $E_1'\cup\id$ with $E_1'$ an
    $\id$-free path expression. First, when $E_2$ is $E_2'\cup\id$
    with $E_2'$ an $\id$-free path expression, then $E$ is equivalent
    to $E_1'\comp E_2'\cup E_1'\cup E_2'\cup \id$ which is of the form
    $E'\cup\id$ with $E'$ the $\id$-free path expression
    $E_1'\comp E_2'\cup E_1'\cup E_2'$. Second, when $E_2$ is
    $\id$-free, $E$ is equivalent to $E_1'\comp E_2 \cup E_2$ which is
    $\id$-free.

    Finally, consider the two case where $E_1$ is $\id$-free and $E_2$
    is $E_2'\cup\id$ with $E_2'$ an $\id$-free path expression, then
    $E$ is equivalent to $E_1\comp E_2'\cup E_1$ which is $\id$-free.
    
  \item $E$ is $E_1\cup E_2$. This case follows immediately by
    induction.
  \item $E$ is $E_1^*$. There are three cases. When $E_1$ is $\id$,
    $E$ is equivalent to $\id$. When $E_1$ is $E_1'\cup\id$ with
    $E_1'$ an $\id$-free path expression, then $E$ is equivalent to
    $E_1'^*$ and clearly $E_1'^*$ is $\id$-free. Lastly, when $E_1$ is
    $id$-free, clearly $E$ is as well. \qedhere
  \end{itemize}
\end{proof}

We also need the notion of ``safe'' path expressions together with the
following Lemma, detailing how path expressions can behave on the
nodes outside a graph.  One can divide all path expressions into the
``safe'' and the ``unsafe'' ones.

\begin{defi}[Safety]
  A path expression is \emph{safe} if one of the following conditions
  holds:
  \begin{itemize}
  \item $E$ is $p$ or $p^-$ with $p$ a property name
  \item $E$ is $E_1\cup E_2$ and both $E_1$ and $E_2$ are safe
  \item $E$ is $E_1\comp E_2$ and at least one of $E_1$ or $E_2$ is safe 
  \end{itemize}
  Otherwise, $E$ is \emph{unsafe}.
\end{defi}

\begin{lem}
  \label{lem:safety}
  Let $E$ be an $\id$-free path expression and let $G$ be a graph.
  \begin{itemize}
  \item If $E$ is safe, then $\iexpr{E}{G} \subseteq N_G \times N_G$.
  \item If $E$ is unsafe, then
    $\iexpr{E}{G} = (\iexpr{E}{G} \cap N_G \times N_G) \cup \{(a,a)
    \mid a\in N - N_G\}$.
  \end{itemize}
\end{lem}

\begin{proof}
By induction. If $E$ is a property name or its inverse, then the
claim clearly holds.  Now assume $E$ is of the form $E_1 \cup
E_2$.  The cases where both $E_1$ and $E_2$ are safe, or both are
unsafe, are clear by induction. If $E_1$ is safe but $E_2$ is
not, then $\iexpr{E}{G} = \iexpr{E_1}{G} \cup \iexpr{E_2}{G} =
(\iexpr{E_1}{G} \cap N_G \times N_G) \cup (\iexpr{E_2}{G} \cap
N_G\times N_G) \cup \{(a,a) \mid a\in N - N_G\} = \iexpr{E}{G}
\cap (N_G \times N_G) \cup \{(a,a) \mid a\in N - N_G\}$. The same
reasoning can be used when $E_2$ is safe but $E_1$ is not.

Next, assume $E$ is of the form $E_1 \comp E_2$.  Furthermore
assume $E_1$ is safe, so that $E$ is safe. 
Let $(x,y) \in \iexpr{E}{G}$.  Then there exists $z$
    such that $(x,z) \in \iexpr{E_1}{G}$ and $(z,y) \in
    \iexpr{E_2}{G}$.  Since $E_1$ is safe, $x$ and $z$ are in
    $N_G$. Now regardless of whether $E_2$ is safe or not,
since $(z,y) \in \iexpr{E_2}{G}$ and $z \in N_G$, we get $y \in
N_G$ as desired.  The same reasoning can be used when $E_2$ is safe.

If $E$ is not safe, we verify that $\iexpr{E}{G} =
(\iexpr{E}{G} \cap N_G \times N_G) \cup \{(a,a) \mid a\in N -
N_G\}$.  For the inclusion from left to right, take $(x,y) \in
\iexpr{E}{G}$. Then there exists $z$ such that $(x,z) \in
\iexpr{E_1}{G}$ and $(z,y) \in \iexpr{E_2}{G}$.  By induction,
there are four cases.  If both $(x,z)$ and $(z,y)$ are in
$N_G\times N_G$, then clearly $(x,y) \in \iexpr{E}{G}\cap
N_G\times N_G$. If both $(x,z),(z,y)$ are in $\{(a,a) \mid a\in N
- N_G\}$ clearly $(x,y) \in \{(a,a) \mid a\in N - N_G\}$. Lastly,
the two cases where one of $(x,z)$ and $(z,y)$ is in $N_G\times
N_G$ and the other in $\{(a,a) \mid a\in N - N_G\}$, are not
possible.

For the inclusion from right to left, take $(x,y) \in
\iexpr{E}{G} \cap (N_G\times N_G) \cup \{(a,a) \mid a\in N -
N_G\}$. If $(x,y) \in \iexpr{E}{G} \cap N_G\times N_G$ then
$(x,y) \in \iexpr{E}{G}$. Otherwise, $(x,y) = (a,a)$ for some $a
\in N-N_G$. Then $(a,a) \in \iexpr{E_1}{G}$ and $(a,a) \in
\iexpr{E_2}{G}$ since $E_1$ and $E_2$ are not safe. We conclude
$(a,a) \in \iexpr{E_1\comp E_2}{G}$ as desired.  

Next, assume $E$ is of the form $E_1^*$. Note that $E$ is unsafe.
By definition of Kleene star, we only need to verify that
$\iexpr{E}{G} \subseteq (\iexpr{E}{G} \cap N_G \times N_G) \cup
\{(a,a) \mid a\in N - N_G\}$.
Let $(x,y) \in \iexpr{E}{G}$. If $x = y$, the claim clearly
holds.  Otherwise, we consider two cases:

\begin{itemize}
\item If $E_1$ is safe, we know
  $\iexpr{E_1}{G} \subseteq N_G\times N_G$. Clearly the
  reflexive-transitive closure of a subset of $N_G\times N_G$ is also
  a subset of $N_G\times N_G$. Therefore, $(x,y)\in N_G\times N_G$ as
  desired.
\item If $E_1$ is unsafe, then by induction
  $\iexpr{E_1}{G} = (\iexpr{E_1}{G} \cap N_G \times N_G) \cup \{(a,a)
  \mid a\in N - N_G\}$. As $x\neq y$ we know $(x,y)$ is in the
  reflexive-transitive closure of $\iexpr{E_1}{G} \cap N_G \times N_G$
  which is a subset of $N_G\times N_G$. \qedhere
\end{itemize}
\end{proof}

Lastly, we define the notion of a \emph{string}, together with the
following Lemma, detailing a convenient property of path expressions.

\begin{defi}
  A string $s$ is a path expression of the form: $\id$, or $s'\comp p$
  or $s'\comp p^-$ where $s'$ is a string and $p$ is a property name.
\end{defi}

\begin{lem}
  \label{lem:strings}
  For every path expression $E$ and every natural number $n$, there
  exists a finite non-empty set of strings $U$ s.t. for every graph
  $G$ with at most $n$ nodes we have
  $\iexpr{E}{G}=\bigcup_{s\in U}\iexpr{s}{G}$.
\end{lem}

The proof of Lemma~\ref{lem:strings} can be found in the appendix.

\subsection{Disjointness}
\label{secdisjproof}

We present here the proof for $X=\disj$.  The general strategy is to
first characterize the behavior of path expressions on $G$ and $G'$.
Then the Proposition is proven with a stronger induction hypothesis,
to allow the induction to carry through.  A similar strategy is
followed in the proof for $X=\eq$.

We begin by defining the graphs $G$ and $G'$ more formally.
\begin{defi}[$G_\disj(\Sigma,m)$]
  \label{def:graph_disj} 
  Let $\Sigma$ be a finite vocabulary including $r$, and let $m$ be a
  natural number. We define the graph $G_\disj(\Sigma,m)$ over the set
  of property names in $\Sigma$ as follows. Let $M=\max(m,3)$. There
  are $4M$ nodes in the graph, which are chosen outside of
  $\Sigma$. We denote these nodes by $x_i^j$ for $i=1,2,3,4$ and
  $j = 1,\dots,M$.  (In the description that follows, subscripts range
  from $1$ to $4$ and superscripts range from $1$ to $M$.)  For each
  property name $p$ in $\Sigma$, the graph has the same set of
  $p$-edges. We describe these edges next. There is an edge from
  $x_i^j$ to $x_{i\bmod 4 + 1}^{j'}$ for every $i$, $j$ and
  $j'$. Moreover, if $i$ is 2 or 4, there is an edge from $x_i^j$ to
  $x_i^{j'}$ for all $j \neq j'$. So, formally, we have:
  $G_\disj(\Sigma,m) := \{ (x_i^j,p,x_{i \bmod 4 + 1}^{j'}) \mid i \in
  \{1,\dots,4\} \;\text{and}\; j,j’ \in \{1,\dots,M\} \;\text{and}\; p
  \in \Sigma \cap P\} \cup \{(x_i^j,p,x_i^{j'}) \mid i \in
  \{1,\dots,4\} \;\text{and}\; j,j’ \in \{1,\dots,M\} \;\text{and}\;
  j\neq j’ \;\text{and}\; p \in \Sigma \cap P\} $.

\end{defi}

Thus, in Figure~\ref{figraphs}, bottom left, one can think of the left oval as
the set of nodes $x_1^j$; the top cloud as the set of nodes
$x_2^j$; and so on.  We call the nodes $x_i^j$ with $i=2,4$ the
\emph{even nodes}, and the nodes $x_i^j$ with $i=1,3$ the
\emph{odd nodes}.

\begin{defi}[$G'_\disj(\Sigma,m)$]
  \label{def:graph_disj_bool}
We define the graph $G'_\disj(\Sigma,m)$
in the same way as $G_\disj(\Sigma,m)$ except
that there is an edge from $x_i^j$ to $x_i^{j'}$ for all $i$ and
$j \neq j'$ (not only for even $i$ values).
\end{defi}

We characterize the behavior of path expressions on the graph
$G_\disj(\Sigma,m)$ as follows.

\begin{lem}
  \label{lem:disjoint_g}
  Let $G$ be $G_\disj(\Sigma,m)$.  Call a path expression
  \emph{simple} if it is a union of expressions of the form
  $s_1\comp\dots\comp s_n$, where $n\geq 1$ and one of the $s_i$ is a
  property name while the other $s_i$ are ``$\id$''.  Let $E$ be a
  non-simple, $\id$-free path expression over $\Sigma$.  The following
  three statements hold:
  \begin{enumerate}
  \item
    \begin{enumerate}[(A)]
    \item \label{caseA} for all even nodes $v$ of $G$, we have
      $\iexpr{E}{G}(v) \supseteq \iexpr{r}{G}(v)$; or
    \item \label{caseB} for all even nodes $v$ of $G$, we have
      $\iexpr{E}{G}(v) \supseteq \iexpr{r^-}{G}(v)$.
    \end{enumerate}
  \item
    \begin{enumerate}[(A)]
      \setcounter{enumii}{2}
    \item \label{caseC} for all odd nodes $v$ of $G$, we have $\iexpr{E}{G}(v) \supseteq \iexpr{r}{G}(v)$; or
    \item \label{caseD} for all odd nodes $v$ of $G$, we have
      $\iexpr{E}{G}(v) \supseteq \iexpr{r^-}{G}(v)$.
    \end{enumerate}
  \item
    \label{item3}
    For all nodes $v$ of $G$, we have $\iexpr{E}{G}(v) - \iexpr{r}{G}(v)
    \neq \emptyset$.
  \end{enumerate}
\end{lem}

\begin{proof}
  For $i=1,2,3,4$, define the \emph{i-th blob of nodes} to be the set
  $X_i=\{x_i^1,\dots,x_i^M\}$ (see Figure~\ref{figraphs}).  We also
  use the notations $\nxt(1)=2$; $\nxt(2)=3$; $\nxt(3)=4$;
  $\nxt(4)=1$; $\prev(4)=3$; $\prev(3)=2$; $\prev(2)=1$; $\prev(1)=4$.
  Thus $\nxt(i)$ indicates the next blob in the cycle, and $\prev(i)$
  the previous.

  The proof is by induction on the structure of $E$. If $E$ is a
  property name, $E$ is simple so the claim is trivial. If $E$ is of
  the form $p^-$, cases B and D are clear and we only need to verify
  the third statement.  That holds because for any $i$, if $v\in X_i$,
  then $\iexpr{p^-}{G}(v)\supseteq X_{\prev(i)}$ and clearly
  $X_{\prev(i)} - \iexpr{r}{G}(v)\neq\emptyset$.  We next consider the
  inductive cases.

  First, assume $E$ is of the form $E_1 \cup E_2$.  When at least one
  of $E_1$ and $E_2$ is not simple, the three statements immediately
  follow by induction, since $\sem E \supseteq \sem{E_1}$ and
  $\sem E \supseteq \sem{E_2}$.  If $E_1$ and $E_2$ are simple, then
  $E$ is simple and the claim is trivial.

  Next, assume $E$ is of the form $E_1^*$.  If $E_1$ is not simple,
  the three statements follow immediately by induction, since
  $\sem E \supseteq \sem{E_1}$.  If $E_1$ is simple, cases A and C
  clearly hold for $E$, so we only need to verify the third statement.
  That holds because, by the form of $E$, every node $v$ is in
  $\iexpr{E}{G}(v)$, but not in $\iexpr{r}{G}(v)$, as $G$ does not
  have any self-loops.

  Finally, assume $E$ is of the form $E_1 \comp E_2$.  Note that if
  $E_1$ or $E_2$ is simple, clearly cases A and C apply to them.  The
  argument that follows will therefore also apply when $E_1$ or $E_2$
  is simple.  We will be careful not to apply the induction hypothesis
  for the third statement to $E_1$ and $E_2$.

  We first focus on the even nodes, and show the first and the third
  statement.  We distinguish two cases.

  \begin{itemize}
  \item If case~A applies to $E_2$, then we show that case~A
    also applies to $E$.  Let $v \in X_i$ be an even node.
    We verify the following two inclusions:
    \begin{itemize}

    \item $\iexpr{E}{G}(v) \supseteq X_i$. Let $u\in X_i$. If $u\neq
      v$, choose a third node $w\in X_i$. Since $X_i$ is a clique,
      $(v,w)\in \iexpr{E_1}{G}$ regardless of whether case A or B
      applies to $E_1$. By case~A for $E_2$, we
      also have $(w,u)\in \iexpr{E_2}{G}$, whence $u\in
      \iexpr{E}{G}(v)$ as desired. If $u=v$, we similarly have $(v,w)
      \in \iexpr{E_1}{G}$ and $(w,u) \in \iexpr{E_2}{G}$ as desired.

    \item $\iexpr{E}{G}(v) \supseteq X_{\nxt(i)}$. Let $u\in
      X_{\nxt(i)}$ and choose $w\neq v \in X_i$. Regardless of
      whether case A or B applies to $E_1$, we have $(v,w) \in
      \iexpr{E_1}{G}$. By case~A for $E_2$, we also have $(w,u)\in \iexpr{E_2}{G}$, whence $u\in \iexpr{E}{G}(v)$ as desired.
    \end{itemize}
    We conclude that $\iexpr{E}{G}(v) \supseteq X_i \cup
    X_{\nxt(i)} \supseteq \sem r$ as desired.

  \item If case~B applies to $E_2$, then we show that case~B
    also applies to $E$. This is analogous to the previous case,
    now verifying that $\iexpr{E}{G}(v) \supseteq X_i \cup X_{\prev(i)}$.
  \end{itemize}

  In both cases, the third statement now follows for even nodes $v$.
  Indeed, $v\in X_i \subseteq \iexpr{E}{G}(v)$ but
  $v\notin \iexpr{r}{G}(v)$.

  We next focus on the odd nodes, and show the second and the third
  statement.  We again consider two cases.
  \begin{itemize}
  \item If case~C applies to $E_1$, then we show that case~C also
    applies to $E$. Let $v \in X_i$ be an odd node.  Note that
    $\iexpr{r}{G}(v) = X_{\nxt(i)}$.  To verify that
    $\iexpr{E}{G}(v) \supseteq X_{\nxt(i)}$, let $u\in
    X_{\nxt(i)}$. Then $u$ is even. Choose $w\neq u\in
    X_{\nxt(i)}$. Since case~C applies to $E_1$, we have
    $(v,w) \in \iexpr{E_1}{G}$.  Moreover, since $X_{\nxt(i)}$ is a
    clique, $(w,u)\in \iexpr{E_2}{G}$ regardless of whether case A or
    B applies to $E_2$. We obtain $(v,u) \in \iexpr{E}{G}$ as desired.

    \bigskip We also verify the third statement for odd nodes in this
    case.  We distinguish two further cases.
    \begin{itemize}
    \item If case~A applies to $E_2$, any node
      $u \in X_{\nxt(\nxt(i))}$ belongs to $\iexpr{E}{G}(v)$, and
      clearly these $u$ are not in $X_{\nxt(i)}=\sem r(v)$.
    \item If case~B applies to $E_2$, then, since $X_i$ is a clique,
      any node $u\in X_i$ belongs to $\iexpr{E}{G}(v)$, and again
      these $u$ are not in $X_{\nxt(i)}$.
    \end{itemize}
  \item If case~D applies to $E_1$, then we show that case~D also
    applies to $E$.  This is analogous to the previous case, now
    verifying that $\iexpr{E}{G}(v) \supseteq X_{\prev(i)}$.  In this
    case the third statement for odd nodes is clear, as clearly
    $X_{\prev(i)} - X_{\nxt(i)} \neq \emptyset$. \qedhere
  \end{itemize}
\end{proof}

We similarly characterize the behavior of path expressions on the
other graph.

\begin{lem}
  \label{lem:disjoint_gprime}
  Let $G'$ be $G'_\disj(\Sigma,m)$ and let $E$ be a non-simple,
  $\id$-free path expression over $\Sigma$. The following statements
  hold:
  \begin{enumerate}
  \item
    $\iexpr{E}{G'} \supseteq \iexpr{r}{G'}$ or 
    $\iexpr{E}{G'} \supseteq \iexpr{r^-}{G'}$.
  \item
    For all nodes $v$ of $G'$,
    we have $\iexpr{E}{G'}(v) - \iexpr{r}{G'}(v) \neq \emptyset$.
  \end{enumerate}
\end{lem}

\begin{proof}
  The proof is similar to the proof of Lemma \ref{lem:disjoint_g}, but
  simpler due to the homogeneous nature of the graph $G'$. We omit the
  proof. 
\end{proof}

We are now ready to prove the non-obvious part of
Proposition~\ref{deprop} where $X=\disj$. We use the following version
of the proposition.

\begin{prop}
  \label{prop:disj}
  Let $V$ be the common set of nodes of the graphs
  $G=G_\disj(\Sigma,m)$ and $G'=G'_\disj(\Sigma,m)$.  Let $\phi$ be a
  shape over $\Sigma$ that does not use $\disj$, and that
  counts to at most $m$. Then either $\iexpr{\phi}{G} \cap V = \emptyset$ or
  $\iexpr{\phi}{G} \supseteq V$.  Moreover,
  $\iexpr{\phi}{G} = \iexpr{\phi}{G'}$.
\end{prop}

\begin{proof}
  This is proven by induction on the structure of $\phi$.  Let $H$ be
  $G$ or $G'$.  If $\phi$ is $\top$, then
  $\iexpr{\top}{H} = N \supseteq V$.  If $\phi$ is $\const c$, then
  $\iexpr{\{c\}}{H} = \{c\} \subseteq \Sigma$ and we know that
  $\Sigma \cap V = \emptyset$.  Next assume $\phi$ is of the form
  $\eq(E, p)$. Using Lemma~\ref{lem:id}, we distinguish four different
  cases for $E$.
  \begin{itemize}
  \item $E$ is $\id$. According to Lemma~\ref{lem:disjoint_g} and
    Lemma~\ref{lem:disjoint_gprime} $\iexpr{E}{H}$ will always contain
    either $\iexpr{p}{H}$ or $\iexpr{p^-}{H}$. In both cases,
    $\iexpr{E}{H}(v)$ clearly never equals
    $\iexpr{\id}{H}(v)=\{v\}$. Therefore,
    $\iexpr{\phi}{H}\cap V=\emptyset$.
  \item $E$ is $E'\cup \id$ where $E'$ is $\id$-free or $E$ itself is
    $\id$-free and non-simple. Lemmas \ref{lem:disjoint_g} and
    \ref{lem:disjoint_gprime} tell us that
    $\iexpr EH(v) - \iexpr rH(v) \neq \emptyset$ for every $v \in V$.
    Since $\iexpr rH = \iexpr pH$, this means $H,v \nmodels \phi$ for
    $v \in V$, or, equivalently, $\iexpr{\phi}{G} \cap V = \emptyset$.
    To see that, moreover, $\iexpr{\phi}{G} = \iexpr{\phi}{G'}$, it
    remains to show that $G,v \models \phi$ iff $G',v \models \phi$
    for all node names $v \notin V$.
  \item $E$ is $\id$-free and simple. Then $\iexpr EH = \iexpr pH$, so
    clearly $\iexpr{\phi}{H} = N \supseteq V$.
  \end{itemize}
  We still need to show $\iexpr{\phi}{G}=\iexpr{\phi}{G'}$. Clearly,
  $\iexpr pG(v)=\iexpr p{G'}(v)=\emptyset$.  Now by
  Lemma~\ref{lem:safety}, if $E$ is safe, then also
  $\iexpr EG(v)=\iexpr E{G'}(v)=\emptyset$, so $G,v\models\phi$ and
  $G',v\models\phi$.  On the other hand, if $E$ is unsafe, then by the
  same Lemma $\iexpr EG(v)=\iexpr E{G'}(v)=\{v\} \neq \emptyset$, so
  $G,v\nmodels\phi$ and $G',v\nmodels\phi$, as desired.

  As the final base case, assume $\phi$ is of the form $\closed(R)$.
  If $\Sigma$ contains a property name $p$ not in $R$, then
  $\iexpr{\phi}{H}\cap V =\emptyset$, since every node in $H$ has an
  outgoing $p$-edge.  Otherwise, i.e., if $\Sigma \subseteq R$, we
  have $\iexpr{\phi}{H}\supseteq V$, since every node in $H$ has only
  outgoing edges labeled by property names in $\Sigma$.  To see that,
  moreover, $\iexpr{\phi}{G} = \iexpr{\phi}{G'}$, it suffices to
  observe that trivially $H,v \models \phi$ for all node names
  $v \notin V$.

  We next consider the inductive cases. The cases for the boolean
  connectives follow readily by induction. Finally, assume $\phi$ is
  of the form $\geqn{n}{E}{\phi_1}$.  By induction, there are two
  possibilities for $\phi_1$:
  \begin{itemize}
  \item If $\iexpr{\phi_1}{H} \cap V = \emptyset$, then also
    $\iexpr{\phi}{H} \cap V = \emptyset$ since path expressions can
    only reach nodes in some graph from nodes in that graph.

  \item If $\iexpr{\phi_1}{H} \supseteq V$, we distinguish three cases
    using Lemma~\ref{lem:id}. First, when $E$ is $\id$, then if $n=1$,
    $\iexpr{\phi}{H}\supseteq V$. Otherwise, if $n\neq 1$, then
    $\iexpr{\phi}{H}=\emptyset$. Next, when $E$ is $\id$-free or
    $E'\cup\id$ with $E'$ an $\id$-free path expression, it suffices
    to show that $\sharp\iexpr{E'}{H}(v) \geq n$ for all $v\in V$.  By
    Lemmas \ref{lem:disjoint_g} and \ref{lem:disjoint_gprime} we know
    that $\iexpr{E_1}{H}(v)$ contains $\iexpr{r}{H}(v)$ or
    $\iexpr{r^-}{H}(v)$.  Inspecting $H$, we see that each of these
    sets has at least $\max(3, m) \geq n$ elements, as
    desired. Finally, when $E$ is equivalent to an $\id$-free path
    expression or whenever $E$ simply does not use $\id$, the argument
    is analogous to the previous case. 
  \end{itemize}

  In both cases we still need to show that
  $\iexpr{\phi}{G} = \iexpr{\phi}{G'}$.  We already showed that
  $\iexpr{\phi}{G} \supseteq V$ and $\iexpr{\phi}{G'} \supseteq V$, or
  $\iexpr{\phi}{G} \cap V = \emptyset$ and
  $\iexpr{\phi}{G'} \cap V = \emptyset$.  Therefore, towards a proof
  of the equality, we only need to consider the node names not in $V$.

  For the inclusion from left to right, take
  $x \in \iexpr{\phi}{G}-V$.  Since $G,x \models \phi$, there exist
  $y_1$, \dots, $y_n$ such that $(x,y_i) \in \iexpr EG$ and
  $G,y_i \models \phi_1$ for $i=1,\dots,n$.  However, since
  $x \notin V$, by Lemma~\ref{lem:safety}, all $y_i$ must equal $x$.
  Hence, $n=1$ and $(x,x) \in \sem E$ and $G,x \models \phi_1$.  Then
  again by the same Lemma, $(x,x) \in \iexpr E{G'}$, since $G$ and
  $G'$ have the same set of nodes $V$.  Moreover, by induction,
  $G',x \models \phi_1$.  We conclude that $G',x \models \phi$ as
  desired.  The inclusion from right to left is argued symmetrically.
\end{proof}

\subsection{Equality}
\label{sec:equality}

Next, we turn our attention to Proposition~\ref{deprop} for
$X=\eq$. We define the graphs from
Figure~\ref{figraphs} formally.

\begin{defi}
  Let $\Sigma$ be a finite vocabulary including $r$, and let $m$ be a
  natural number.  Choose a set $V$ of node names outside $\Sigma$, of
  cardinality $M:=\max(3,m+1)$.  Fix two arbitrary nodes $a$ and $b$
  from $V$.  We define the graph $G_\eq(\Sigma)$ over the set of
  property names from $\Sigma$ as follows. For each property name $p$
  in $\Sigma$, the set of $p$-edges in $G_\eq(\Sigma)$ equals
  $V\times V - (b,a)$.  We define the graph $G_\eq'(\Sigma)$
  similarly, but with $V \times V$ as the set of $p$-edges.
\end{defi}

So, $G'_\eq(\Sigma,m)$ is a complete graph, and $G_\eq(\Sigma,m)$ is a
complete graph with one edge $(b,a)$ removed.

\begin{lem}
  \label{lem:eq_notin_self}
  Let $E$ be an $\id$-free path expression over $\Sigma$ and let
  $H=G_\eq(\Sigma,m)$ or $G'_\eq(\Sigma,m)$.
  Then
  \begin{enumerate}[A.]
  \item $\iexpr{E}{H} \supseteq \iexpr{r}{H}$, or
  \item $\iexpr{E}{H} \supseteq \iexpr{r^-}{H}$.
  \end{enumerate}
\end{lem}

\begin{proof}
  The claim is obvious for $G'_\eq(\Sigma,m)$, being a complete graph.
  So we focus on the graph $G_\eq(\Sigma,m)$.  The proof is by
  induction.  If $E$ is a property name or its inverse, the claim is
  clear.  If $E$ is of the form $E_1 \cup E_2$, the claim is immediate
  by induction.

  Assume $E$ is of the form $E_1 \comp E_2$.  We show that A
  applies.\footnote{Actually, $\sem E$ always contains $V\times V$ in
    this case, but we do not need this.} If A applies to $E_1$, this
  is clear, since we can follow any edge by $E_1$ and then stay at the
  head of the edge by $E_2$ using the self-loop.  If B applies to
  $E_1$, the same can still be done for all edges except for $(a,b)$,
  which is the only nonsymmetrical edge.  To go from $a$ to $b$ by
  $E$, we go by $E_1$ from $a$ to a node $c$ distinct from $a$ and
  $b$, then go by $E_2$ from $c$ to $b$.

  If $E$ is of the form $E_1^*$, again A applies, since ${E_1^*}$
  contains ${E_1/E_1}$.
\end{proof}

We are now ready to prove the non-obvious part of
Proposition~\ref{deprop} where $X=\eq$. We use the following version
of the proposition.

\begin{prop}
  \label{prop:eq}
  Let $G$ be $G_\eq(\Sigma,m)$ and let $G'$ be $G'_\eq(\Sigma,m)$.  Let
  $\phi$ be a shape over $\Sigma$ that does not use $\eq$ and that
  counts to at most $m$. Then either
  $\iexpr{\phi}{G} \cap V = \emptyset$ or
  $\iexpr{\phi}{G} \supseteq V$. Moreover,
  $\iexpr{\phi}{G} = \iexpr{\phi}{G'}$.
\end{prop}

\begin{proof}
  This is proven by induction on the structure of $\phi$.  Let $H$ be
  $G$ or $G'$. We focus directly on the relevant cases. Assume $\phi$
  is of the form $\disj(E_1,E_2)$.  Lemma~\ref{lem:eq_notin_self}
  clearly yields that $\iexpr{\phi}{H} \cap V = \emptyset$.  It again
  remains to verify that $G,v \models \phi$ iff $G',v \models \phi$
  for all node names $v \notin V$.  By Lemma~\ref{lem:safety}, for
  such $v$ and $H = G$ or $G'$, we indeed have $H,v \models \phi$ if
  exactly one of $E_1$ and $E_2$ is safe.  If both are safe or both
  are unsafe, we have $H,v \nmodels \phi$.

  The last base case of interest is the case where $\phi$ is of the
  form $\closed(R)$. This goes again exactly as in the proof for
  $X=\disj$.

  We next consider the inductive cases. The cases for the boolean
  connectives follow readily by induction. Finally, assume $\phi$ is
  of the form $\geqn{n}{E}{\phi_1}$.  By induction, there are two
  possibilities for $\phi_1$:
  \begin{itemize}
  \item If $\iexpr{\phi_1}{G} \cap V = \emptyset$ then
    $\iexpr{\phi}{G} \cap V = \emptyset$ since path expressions can
    only reach nodes in some graph from nodes in that graph.
  \item If $\iexpr{\phi_1}{H} \supseteq V$, we distinguish three
    cases using Lemma~\ref{lem:id}. First, when $E$ is $\id$, then if
    $n=1$, $\iexpr{\phi}{H}\supseteq V$. Otherwise, if $n\neq 1$, then
    $\iexpr{\phi}{H}=\emptyset$. Next, when $E$ is $\id$-free or
    $E'\cup\id$ with $E'$ an $\id$-free path expression, it suffices
    to show that $\sharp\iexpr{E'}{H}(v) \geq n$ for all $v\in V$.
    By Lemma~\ref{lem:eq_notin_self}, we know that $\iexpr EH(v)$
    contains $\iexpr rH(v)$ or $\iexpr {r^-}H(v)$.  These sets contain
    at least $M - 1 \geq m \geq n$ elements as desired.  (The number
    $M-1$ is reached only when $H$ is $G$ and $v=b$ or $v=a$;
    otherwise the sets contain $M$ elements.)
  \end{itemize}

  The equality $\iexpr{\phi}{G} = \iexpr{\phi}{G'}$ is shown in the
  same way as in the proof for $X=\disj$
  (Section~\ref{secdisjproof}).
\end{proof}

\subsection{Closure}
\label{seclosure}

Without using $\closed$, shapes cannot say anything about properties
that they do not explicitly mention.  We formalize this intuitive
observation as follows.  The proof is straightforward.

\begin{lem}
  \label{lemclos} Let $\Sigma$ be a vocabulary, let $E$ be a path
  expression over $\Sigma$, and let $\phi$ be a shape over $\Sigma$
  that does not use $\closed$.  Let $G_1$ and $G_2$ be graphs such
  that $\iexpr p{G_1} = \iexpr p{G_2}$ for every property name $p$ in
  $\Sigma$.  Then $\iexpr E{G_1} = \iexpr E{G_2}$ and
  $\iexpr \phi{G_1} = \iexpr \phi{G_2}$. \qed
\end{lem}

Theorem~\ref{bool} now follows readily for $X = \closed$.  Let $F$ be
a feature set without $\closed$, let $\Sch$ be a shape schema in
$\lang(F)$, and let $\phi$ be the validation shape of $\Sch$.  Let $p$
be a property name not mentioned in $\Sch$, and different from $r$.
Consider the graphs $G = \{(a,r,a),(a,p,a)\}$ and $G'=\{(a,r,a)\}$, so
that $G'$ belongs to $Q_\closed$ but $G$ does not.  By
Lemma~\ref{lemclos} we have $\iexpr\phi G = \iexpr\phi {G'}$, showing
that $\Sch$ does not define $Q_\closed$.

\begin{rem}
  Lemma~\ref{lemclos} fails completely in the presence of closure
  constraints.  The simplest counterexample is to consider
  $\Sigma=\emptyset$ and the shape $\closed(\emptyset)$.  Trivially,
  any two graphs agree on the property names from $\Sigma$. However,
  $\sem{\closed(\emptyset)}$, which equals the set of node names that
  do not have an outgoing edge in $G$ (they may still have an incoming
  edge), obviously depends on the graph $G$.
\end{rem}

\newcommand{\iexpra}[2]{\iexpr{#1}{#2}_{\mathit{adom}}}

The reader may wonder if this statement still holds under active
domain semantics. In such semantics, which we denote by
$\iexpra{\phi}{G}$, we would view $G$ as an interpretation with domain
\emph{not} the whole of $N$; rather we would take as domain the set
$N_G\cup C$, with $C$ the set of constants mentioned in $\phi$. When
assuming active domain semantics, a modified lemma is required. To see
this, consider the graph $G= \{(a,p,b)\}$ and
$G'=\{(a,p,b), (a,q,c)\}$. Let $\phi$ simply be $\top$. We
have $\iexpra{\phi}{G}=\{a, b\}$ and $\iexpra{\phi}{G'}=\{a,b,c\}$, so
Lemma~\ref{lemclos} no longer holds. We can, however, give the following more refined variant of Lemma~\ref{lemclos}:

\begin{lem}
  \label{lemclosadom}
  Let $\Sigma$ be a vocabulary, let $E$ be a path expression over
  $\Sigma$, and let $\phi$ be a shape over $\Sigma$ that does not use
  $\closed$.  Let $I_1$ and $I_2$ be interpretations such that
  $\iexpr p{I_1} = \iexpr p{I_2}$ for every property name $p$ in
  $\Sigma$.  Then
  $\iexpr E{I_1} \cap \Delta^{I_2}\times \Delta^{I_2}= \iexpr E{I_2} \cap
  \Delta^{I_1}\times \Delta^{I_1}$ and
  $\iexpr \phi{I_1}\cap \Delta^{I_2} = \iexpr \phi{I_2} \cap
  \Delta^{I_1}$. \qed
\end{lem}

The same reasoning as given after Lemma~\ref{lemclos}, now using the
new Lemma, then shows that $\closed$ is still primitive under active
domain semantics.

\section{Are target-based shape schemas enough?}
\label{sec:targetres}
Lemma~\ref{lemclos} also allows us to clarify that, as far as
expressive power is concerned, and in the absence of closure
constraints, the restriction to target-based shape schemas is
inconsequential.

\begin{thm}
  \label{thm:generalizedschema}
  Every generalized shape schema that does not use
  closure constraints is equivalent to a target-based shape schema
  (that still does not use closure constraints).
\end{thm}

In order to prove this theorem, we first establish the following
lemma.

\begin{lem}
\label{lembuiten}
Let $\phi$ be a shape and let $C$ be the set of constants mentioned in
$\phi$.  Assume there exists a graph $G$ and a node name
$x \notin N_G \cup C$ such that $G,x \models \phi$.  Then for any graph
$H$ and any node name $y \notin N_H \cup C$, also $H,y \models \phi$.
\end{lem}
\begin{proof}
By induction on $\phi$.  The case where $\phi$ is of the form
$\const c$ cannot occur, and the case where $\phi$ is $\top$ is
trivial.

If $\phi$ is $\phi_1 \lor \phi_2$ or $\neg \phi_1$,
the claim follows readily by induction.

Now assume $\phi$ is of the form $\geqn nE{\phi_1}$.  Then there
exists $x_1,\dots,x_n$ such that $(x,x_i) \in \sem E$ and
$G,x_i \models \phi_1$ for $i=1,\dots,n$.  However, since
$x \notin N_G$, by Lemma~\ref{lem:safety}, all $x_i$ must equal $x$.
Hence, $n=1$ and $(x,x) \in \sem E$ and $G,x \models \phi_1$.  By the
same Lemma, $(y,y) \in \iexpr EH$, since $y\notin N_H$.  Furthermore,
by induction, $H,y \models \phi_1$.  We conclude that
$H,y \models \phi$ as desired.

Next, assume $\phi$ is $\eq(E,p)$.  Since $G,x \models \phi$, but
$\sem p = \emptyset$ since $x \notin N_G$, also
$\sem E(x) = \emptyset$.  Then by Lemma~\ref{lem:safety}, also
$\iexpr EH(y) = \emptyset$, since $y \notin N_H$.  Furthermore, also
$\iexpr pH(y) = \emptyset$.  We conclude that $H,y \models \phi$ as
desired.

Next assume $\phi$ is $\disj(E,p)$.  Then $H,y \models \phi$
is clear.  Indeed, since $y \notin N_H$, we have $\iexpr
pH(y)=\emptyset$.

Finally, assume $\phi$ is $\closed(R)$.  Then again $H,y \models \phi$
is clear because $y \notin N_H$.
\end{proof}

We can now show the theorem.

\begin{proof}[Proof of Theorem~\ref{thm:generalizedschema}]
  Let $\phi$ be the validation shape for shape schema $\Sch$, so that
  $G \models \Sch$ if and only if $\sem\phi$ is empty.  Let $C$ be the
  set of constants mentioned in $\phi$.

  Let us say that $\phi$ is \emph{internal} if for every graph $G$ and
  every node name $v$ such that $G,v \models \phi$, we have
  $v \in N_G \cup C$.  If $\phi$ is not internal, then, using
  Lemma~\ref{lembuiten}, for every graph $G$ and every node
  $v \notin N_G \cup C$, we have $G,v \models \phi$.  Thus, if $\phi$
  is not internal, $\Sch$ is unsatisfiable and is equivalent to the
  single target-based inclusion $\const c \subseteq \neg \top$, for an
  arbitrary constant $c$.

  So now assume $\phi$ is internal.  Define the target-based shape
  schema $\mathcal T$ consisting of the following inclusions:
  \begin{itemize}
  \item For each constant $c \in C$, the inclusion
    $\const c \subseteq \neg \phi$.
  \item For each property name mentioned in $\phi$, the two inclusions
    $\exists p.\top \subseteq \neg \phi$ and
    $\exists p^-.\top \subseteq \neg \phi$.
  \end{itemize}
  We will show that $\Sch$ and $\mathcal T$ are equivalent.  Let
  $\psi$ be the validation shape for $\mathcal T$.

  Let $G$ be any graph, and let $G'$ be the graph obtained from $G$ by
  removing all triples involving property names not mentioned in
  $\phi$.  We reason as follows:
  \begin{align*}
    G \models \Sch
    &\Leftrightarrow
      \sem\phi = \emptyset \\
    &\Leftrightarrow
      \iexpr\phi{G'} = \emptyset && \text{by Lemma~\ref{lemclos}} \\
    &\Leftrightarrow
      G' \models \mathcal T && \text{since $\phi$ is internal} \\
    &\Leftrightarrow
      \iexpr\psi{G'} = \emptyset \\
    &\Leftrightarrow
      \iexpr\psi{G} = \emptyset && \text{by Lemma~\ref{lemclos}} \\
    &\Leftrightarrow
      G \models \mathcal T 
    \tag*{\qedhere}
  \end{align*} 
\end{proof}

\begin{rem}
Note that we do not need class-based targets in the proof, so such
targets are redundant on the left-hand sides of inclusions.
This can also be seen directly: any inclusion
\[ \exists \textsf{type}/\textsf{subclass}^*.\const c \subseteq
\phi \] with a class-based target is equivalent to the following
inclusion with a
subjects-of target: \[ \exists \textsf{type}.\top
\subseteq \neg
\exists \textsf{type}/\textsf{subclass}^*.\const c \lor \phi \]
\end{rem}

\begin{rem}
  Theorem~\ref{thm:generalizedschema} fails in the presence of closure
  constraints.  For example, the inclusion
  $\neg \closed(\emptyset) \subseteq \exists r.\top$ defines the class
  of graphs where every node with an outgoing edge has an outgoing
  $r$-edge.  Suppose this inclusion would be equivalent to a
  target-based shape schema $\Sch$, and let $R$ be the set of all
  property names mentioned in the targets of $\Sch$.  Let $p$ be a
  property name not in $R$ and distinct from $r$; let $a$ be a node
  name not used as a constant in $\Sch$; and consider the graph
  $G=\{(a,p,a)\}$.  This graph trivially satisfies $\Sch$, but
  violates the inclusion.
\end{rem}

\newcommand{\strset}[1]{\bigcup_{s\in U}\iexpr{s}{#1}}

\section{Extensions for full equality and disjointness tests}
\label{sec:fulltests}
A quirk in the design of SHACL is that it only allows equality and
disjointness tests $\eq(E_1,E_2)$ and $\disj(E_1,E_2)$ where $E_1$ can
be a general path expression, but $E_2$ needs to be a property name.
The next question we can ask is whether allowing ``full'' equality or
disjointness tests, i.e., allowing a general path expression for
$E_2$, strictly increases the expressive power.
Within the community there are indeed plans to extend SHACL in
this direction \cite{dash,shacl_extensions}.

When we allow for such ``full'' equality and disjointness tests, it
gives rise to two new features: $\fulleq$ and $\fulldisj$.
Formally, we extend the grammar of shapes with two new constructs:
$\eq(E_1,E_2)$ and $\disj(E_1,E_2)$.

\begin{rem}
  We extend Remark~\ref{rem:id} by noting that in real SHACL, we
  cannot explicitly write the shapes $\eq(\id,\id)$ and
  $\disj(\id,\id)$. However, these shapes are equivalent to $\top$ and
  $\neg\top$ respectively.
\end{rem}

We are going to show that each of these new features strictly adds
expressive power. Concretely, we introduce the following classes of
graphs.
\begin{description}
\item[Full equality] $Q_\fulleq$ is the class of graphs where all
  objects of a property name $p$ do not have the same subjects for $p$
  and $q$. Note that $Q_{\fulleq}$ is definable in $\lang(\fulleq)$ by
  the single, target-based, inclusion statement
  $\exists p^-.\top \subseteq \neg\eq(p^-,q^-)$.
\item[Full disjointness] $Q_\fulldisj$ is the class of graphs where
  all objects of a property name $p$ do not have disjoint sets of
  subjects for $p$ and $q$. Note that $Q_{\fulldisj}$ is definable in
  $\lang(\fulldisj)$ by the single, target-based, inclusion statement
  $\exists p^-.\top \subseteq \neg\disj(p^-,q^-)$.
\end{description}

In the spirit of Theorem~\ref{bool}, we are now going to show the
following:
\begin{thm}
  \label{thm:fulleqdisj}
  $Q_\fulleq$ is not definable in $\lang(\eq, \fulldisj, \closed)$ and
  $Q_\fulldisj$ is not definable in $\lang(\disj, \fulleq, \closed)$.
\end{thm}

These two non-definability results are proven in the following
Sections~\ref{sec:fulleq} and \ref{sec:fulldisj}. Then in
Section~\ref{sec:furth-non-defin} we will reconsider the
non-definability results for non-full equality and disjointness from
Theorem~\ref{bool} in the new light of their full versions.

\subsection{Full equality}
\label{sec:fulleq}

We present here the proof for the primitivity of full equality
tests. The general strategy is the same as in
Section~\ref{secexpr}, where again we will prove appropriate
versions of Proposition~\ref{deprop}.

We begin by defining the graphs $G$ and $G'$ formally.
Note that, as desired, $G'$ belongs to $Q_{\fulleq}$ but $G$ does not.
\begin{defi}{$G_\fulleq(\voc, m)$}
  Let $\voc$ be a finite vocabulary and let $m\geq 3$ be a natural
  number. Let $A=\{a_1,\dots,a_m\}$, $B=\{b_1,\dots,b_m\}$ and
  $C=\{c_1,\dots,c_m\}$ be three disjoint sets of nodes, disjoint
  from $\voc$. We define the graph $G_\fulleq(\voc, m)$ to be
  $\iexpr{p}{G} = C\times (A\cup B)$ and
  $\iexpr{q}{G} = C\times A \cup \{(c_i, b_j)\mid i\neq j \in
  \{1,\dots,m\}\}$.
\end{defi}

\begin{defi}{$G'_\fulleq(\voc, m)$}
  We define the graph $G'_\fulleq(\voc, m)$ like $G_\fulleq(\voc, m)$
  but
  $\iexpr{q}{G} = \{(c_i, a_j)\mid i\neq j \in \{1,\dots,m\}\} \cup
  \{(c_i, b_j)\mid i\neq j \in \{1,\dots,m\}\}$.
\end{defi}

We identify the possible types of strings on the graphs
$G_\fulleq(\Sigma,m)$ and $G'_\fulleq(\Sigma,m)$ as follows.

\begin{lem}
  \label{lem:eq-pe}
  Let $\voc$ be a vocabulary. Let $m \geq 3$ be a natural number. Let
  $G$ be $G_\fulleq(\voc, m)$ and let $G'$ be $G'_\fulleq(\voc,
  m)$. The only possibilities for a string $s$ evaluated on $G$ and
  $G'$ are the following:
  \begin{enumerate}
  \item $\iexpr{s}{G} = \iexpr{p}{G} = \iexpr{s}{G'} = C\times (A\cup B)$.
  \item
    $\iexpr{s}{G} = \iexpr{q}{G} = (C\times A) \cup \{(c_i,b_j)\mid
    i\neq j \in
    \{1,\dots, m\})\}$ and \\
    $\iexpr{s}{G'} = \iexpr{q}{G'} = \{(c_i, a_j)\mid i\neq j \in
    \{1,\dots,m\}\} \cup \{(c_i, b_j)\mid i\neq j \in
    \{1,\dots,m\}\}$.
  \item
    $\iexpr{s}{G} = \iexpr{p^-}{G} = \iexpr{s}{G'} = (A\cup B)\times
    C$.
  \item
    $\iexpr{s}{G} = \iexpr{q^-}{G} = (A\times C) \cup \{(b_i,c_j)\mid
    i\neq j \in
    \{1,\dots, m\})\}$ and \\
    $\iexpr{s}{G'} = \iexpr{q^-}{G'} = \{(a_i,c_j)\mid i\neq j \in
    \{1,\dots, m\})\} \cup \{(b_i,c_j)\mid i\neq j \in \{1,\dots,
    m\})\}$.
  \item $\iexpr{s}{G} = \iexpr{s}{G'} = C\times C$.
  \item $\iexpr{s}{G} = \iexpr{s}{G'} = (A\cup B)\times (A\cup B)$.
  \item $\iexpr{s}{G} = \iexpr{s}{G'} = \id$.
  \item $\iexpr{s}{G} = \iexpr{s}{G'} = \emptyset$.
  \end{enumerate}
\end{lem}

\begin{proof}
  We show this by systematically enumerating all strings until no new binary
  relations can be found. Note that we only enumerate over strings
  that alternate between property names and reversed property
  names. Indeed, all other strings evaluate to the empty relation on
  both $G$ and $G'$. Every time we encounter new binary relations, we
  put the string in boldface.

  \begin{center}
  \begin{tabular}[b]{lcc}
    \toprule
    $s$ & $\iexpr{s}{G}$ & $\iexpr{s}{G'}$ \\
    \midrule
    $\mathbf{id}$ & $\id$ & $\id$ \\ 
    $\mathbf{p}$ & $C\times (A\cup B)$ & $C\times (A\cup B)$ \\
    $\mathbf{q}$ & $(C\times A)\cup{}$ & $\{(c_i, a_j)\mid i\neq j \in \{1,\dots,m\}\}\cup{}$ \\
        & $\{(c_i,b_j)\mid i\neq j \in
          \{1,\dots, m\})\}$ & $\{(c_i, b_j)\mid i\neq j \in \{1,\dots,m\}\}$\\
    $\mathbf{p^-}$ & $(A\cup B)\times C$ & $(A\cup B)\times C$ \\
    $\mathbf{q^-}$ & $(A\times C)\cup{}$ & $\{(a_i,c_j)\mid i\neq j \in \{1,\dots, m\})\}\cup{}$ \\
        & $\{(b_i,c_j)\mid i\neq j \in
          \{1,\dots, m\})\}$ & $\{(b_i,c_j)\mid i\neq j \in \{1,\dots, m\})\}$ \\
    $\mathbf{p\comp p^-}$ & $C\times C$ & $C\times C$ \\
    $p\comp q^-$ & $C\times C$ & $C\times C$ \\
    $q\comp p^-$ & $C\times C$ & $C\times C$ \\
    $q\comp q^-$ & $C\times C$ & $C\times C$ \\
    $\mathbf{p^-\comp p}$ & $(A\cup B)\times (A\cup B)$ & $(A\cup B)\times (A\cup B)$ \\
    $p^-\comp q$ & $(A\cup B)\times (A\cup B)$ & $(A\cup B)\times (A\cup B)$ \\
    $q^-\comp p$ & $(A\cup B)\times (A\cup B)$ & $(A\cup B)\times (A\cup B)$ \\
    $q^-\comp q$ & $(A\cup B)\times (A\cup B)$ & $(A\cup B)\times (A\cup B)$ \\
    $p\comp p^-\comp p$ & $C\times (A\cup B)$ & $C\times (A\cup B)$ \\
    $p\comp p^-\comp q$ & $C\times (A\cup B)$ & $C\times (A\cup B)$ \\
    $p^-\comp p\comp p^-$ & $(A\cup B)\times C$ & $(A\cup B)\times C$ \\
    $p^-\comp p\comp q^-$ & $(A\cup B)\times C$ & $(A\cup B)\times C$ \\
    \bottomrule
  \end{tabular} \qedhere
\end{center} 
\end{proof}

We are now ready to prove the key proposition.

\begin{prop}
  Let $\voc$ be a vocabulary. Let $m\geq 3$ be a natural number. Let
  $V=A\cup B\cup C$ be the common set of nodes of the graphs
  $G = G_\fulleq(\voc, m)$ and $G' = G'_\fulleq(\voc, m)$. For all
  shapes $\phi$ over $\voc$ counting to at most $m-1$, we have
  $\iexpr{\phi}{G} = \iexpr{\phi}{G'}$. Moreover, 
  \begin{itemize}
  \item $\iexpr{\phi}{G} \cap V = A\cup B$, or
  \item $\iexpr{\phi}{G} \cap V = C$, or
  \item $\iexpr{\phi}{G} \cap V = V$, or 
  \item $\iexpr{\phi}{G} \cap V = \emptyset$.
  \end{itemize}
\end{prop}

\begin{proof}
  By induction on the structure of $\phi$. For the base cases, if
  $\phi$ is $\top$ then $\iexpr{\top}{G}=\iexpr{\top}{G'}=N$ and
  $N\cap V = V$. If $\phi$ is $\{c\}$, then
  $\iexpr{\{c\}}{G}=\iexpr{\{c\}}{G'} = \{c\}$ and
  $\{c\}\cap V=\emptyset$ since $c\in\voc$ and $V\cap\voc=\emptyset$.

  If $\phi$ is $\closed(Q)$, we consider the possibilities for $Q$. If
  $Q$ does not contain both $p$ and $q$, then clearly
  $\iexpr{\phi}{G}\cap V=\iexpr{\phi}{G'}\cap V=A\cup B$. Otherwise,
  $\iexpr{\phi}{G}=\iexpr{\phi}{G'}=N$.

  Before considering the remaining cases, we observe the following
  symmetries:
  \begin{itemize}
  \item All elements of $A$ are symmetrical in $G$. This is obvious
    from the definition of $G$.
  \item Also in $G'$, all elements of $A$ are symmetrical. Indeed, for
    any $a_i\neq a_j$ in $A$, the function that swaps $a_i$ and $a_j$,
    as well as $c_i$ and $c_j$, is an automorphism of $G'$.
  \item Similarly, all elements of $B$ are symmetrical in $G$, and
    also in $G'$.
  \item Moreover, we see that all elements of $C$ are symmetrical in
    $G$, and in $G'$.
  \item Finally, in $G'$, any $a_i$ and $b_j$ are symmetrical. Indeed,
    the function that swaps $a_i$ and $b_i$ is clearly an automorphism
    of $G'$. In turn, $b_i$ and $b_j$ are symmetrical by the above.
  \end{itemize}

  Therefore, we are only left to show:
  \begin{enumerate}[(i)]
  \item \label{en:p11} For any $a\in A$ and $b\in B$, we have
    $G,a\models\phi\iff G, b\models \phi$,
  \item \label{en:p12} For any $a\in A$, we have
    $G,a\models\phi\iff G',a\models\phi$, and
  \item \label{en:p13} For any $c\in C$, we have
    $G,c\models\phi\iff G',c\models\phi$.
  \item \label{en:p14} For any $x\not\in V$, we have
    $G,x\models\phi\iff G',x\models\phi$.
  \end{enumerate}
  Note that then also for any $b\in B$, we have
  $G,b\models\phi\iff G',b\models\phi$ because for any $a\in A$ and
  $b\in B$, we have
  $G,b\models\phi\stackrel{\rm\ref{en:p11}}{\iff} G, a\models \phi
  \stackrel{\rm\ref{en:p12}}{\iff} G',a\models\phi
  \stackrel{\mathit{symmetry}}{\iff} G',b\models\phi$.

  Consider the case where $\phi$ is $\eq(E,r)$. We
  verify \ref{en:p11}, \ref{en:p12}, \ref{en:p13}, and \ref{en:p14}.
  \begin{enumerate}
  \item[\ref{en:p11}] By definition of $G$,
    $\iexpr{r}{G}(a)=\iexpr{r}{G}(b)=\emptyset$ for any property name
    $r$. Therefore we need to show
    $\iexpr{E}{G}(a)=\emptyset \iff \iexpr{E}{G}(b)=\emptyset$. By
    Lemma \ref{lem:strings} we know there is a set of strings $U$
    equivalent to $E$ in both $G$ and $G'$. By Lemma \ref{lem:eq-pe}
    there are only 8 types of strings. We observe from Lemma
    \ref{lem:eq-pe} that for every $U$,
    $\strset{G}(a)$ is empty whenever $U$ only
    contains strings of type 1, 2, 5, or 8. These are also exactly the
    $U$ s.t.\ $\strset{G}(b)$ is empty.
  \item[\ref{en:p12}] Furthermore, these are also exactly the sets of
    strings $U$ s.t.\ $\strset{G'}(a)$ is
    empty. Therefore, as $\iexpr{r}{G'}(a)=\emptyset$, we have
    $G',a\models\phi$.
  \item[\ref{en:p13}] Assume $G, c\models\phi$. We consider the possibilities for
    $r$. First, suppose $r=p$. The sets of strings $U$ s.t.
    $\strset{G}(c)=\iexpr{p}{G}(c)$ contain strings
    of type 1 but not strings of type 5 or 7. These are also exactly
    the $U$ s.t.\ $\strset{G'}(c)=\iexpr{p}{G'}(c)$.

    Next, suppose $r=q$. The sets of strings $U$ s.t.
    $\strset{G}(c)=\iexpr{p}{G}(c)$ contain strings
    of type 2 but not strings of type 1, 5 or 7. These are also
    exactly the types of strings s.t.\
    $\strset{G'}(c)=\iexpr{q}{G'}(c)$.

    Finally, if $r$ is any other property name, then
    $\iexpr{r}{G}(c)=\iexpr{r}{G'}(c)=\emptyset$. This is the case
    when $U$ does not contain any strings of type 1, 2, 3, or 7. These
    are also exactly the types of strings $U$ s.t.\
    $\strset{G'}(c)=\emptyset$.
  \item[\ref{en:p14}] Let $x\in N-V$. Clearly
    $\iexpr{r}{G}(x)=\iexpr{r}{G'}(x)=\emptyset$. By Lemma
    \ref{lem:safety}, if $E$ is safe, then
    $\iexpr{E}{G}(x)=\iexpr{E}{G'}(x)=\emptyset$. Therefore
    $G,x\models\phi$ and $G',x\models\phi$. On the other hand,
    whenever $E$ is unsafe,
    $\iexpr{E}{G}(x)=\iexpr{E}{G'}(x)=\{x\}\neq\emptyset$. Therefore,
    $G,x\not\models\phi$ and $G',x\not\models\phi$.
  \end{enumerate}

\noindent
  Next, consider the case where $\phi$ is $disj(E_1,E_2)$. We again
  verify \ref{en:p11}, \ref{en:p12}, \ref{en:p13}, and \ref{en:p14}.
  \begin{enumerate}
  \item[\ref{en:p11}] Assume $G,a\models\disj(E_1,E_2)$.  This can
    only be the case when the corresponding sets of strings $U_1$ and
    $U_2$ are of the following form. $U_1$ can consist only of strings
    of type 3, 4, 1, 2, 5, and 8 (Here, types 1, 2, 5 and 8 evaluate
    to empty from $a$ as already seen above). $U_2$ can then only
    consist of strings of type 6, 7, 1, 2, 5, and 8 (or vice
    versa). These are also the only cases where
    $G,b\models\disj(E_1,E_2)$.
  \item[\ref{en:p12}] Exactly the same situation occurs in $G'$ and
    these are then also the only cases where
    $G',a\models\disj(E_1,E_2)$.
  \item[\ref{en:p13}] Assume $G,c\models\disj(E_1,E_2)$. This can only
    be the case when the corresponding sets of strings $U_1$ and $U_2$
    are of the following form. $U_1$ can consist only of strings of
    type 1, 2, 3, 4, 6, and 8 (Here, types 3, 4, 6, and 8 evaluate to
    empty from $c$ as already seen above). $U_2$ can then only consist
    of strings of type 5, 7, 3, 4, 6, and 8. We observe that this is
    also the case in $G'$.
  \item[\ref{en:p14}] Let $x\in N-V$. Whenever $E_1$ is safe, by Lemma
    \ref{lem:safety}
    $\iexpr{E_1}{G}(x)=\iexpr{E_1}{G'}(x)=\emptyset$. Therefore,
    $G,x\models\phi$ and $G',x\models\phi$. Clearly, the same holds
    whenever $E_2$ is safe. When both $E_1$ and $E_2$ are unsafe,
    $\iexpr{E_1}{G}(x)=\iexpr{E_1}{G'}(x)=\{x\}\neq\emptyset$. Therefore,
    $G,x\not\models\phi$ and $G',x\not\models\phi$.
  \end{enumerate}

  The cases where $\phi$ is $\phi_1\land\phi_2$, $\phi_1\lor\phi_2$ or
  $\neg\phi_1$ are handled by induction in a straightforward manner.

  Lastly, we consider the case where $\phi$ is $\geqn{n}{E}{\phi_1}$.
  \begin{enumerate}
  \item[\ref{en:p11}] Assume $G, a\models\phi$. Then, there exist
    distinct $x_1,\dots,x_n$ s.t.\ $(a,x_i)\in\iexpr{E}{G}$ and
    $G,x_i\models\phi_1$ for $1\leq i\leq n$. Again, by Lemma
    \ref{lem:strings} we know there is a set of strings $U$ equivalent
    to $E$ in both $G$ and $G'$. By Lemma \ref{lem:eq-pe} there are
    only 8 types of strings.  By induction, we consider three cases.

    First, if $\iexpr{\phi_1}{G}\cap V=A\cup B$, then
    $x_i,\dots,x_n\in A\cup B$. Therefore, we know $U$ must at least
    contain strings of type 6 or 7. Suppose $U$ contains strings of
    type 6. Then, we verify that
    $\sharp (\strset{G}(b)\cap (A\cup B)) \geq m-1 \geq n$. Whenever $U$
    contains strings of type 7, and not of type 6, we know $n=1$ and
    clearly $\sharp (\strset{G}(b)\cap (A\cup B))=1$.

    Next, if $\iexpr{\phi_1}{G}\cap V=C$, then $x_i,\dots,x_n\in
    C$. Therefore, we know $U$ must at least contain strings of type 3
    or 4. Whenever $U$ contains 3 or 4, we verify that
    $\sharp (\strset{G}(b)\cap C) \geq m-1 \geq n$.

    Next, if $\iexpr{\phi_1}{G}\cap V=V$, then $x_i,\dots,x_n\in
    V$. Therefore, we know $U$ must at least contain strings of type
    3, 4, 6 or 7. All these types have already been handled in the
    previous two cases.

    Finally, the case where $\iexpr{\phi_1}{G}\cap V=\emptyset$ cannot
    occur as $n>0$.

    The sets of strings $U$ above are also exactly the sets used to
    argue the implication from right to left.
  \item[\ref{en:p12}] For every case of $U$ above, for every inductive
    case of $\phi_1$, we can also verify that
    $\sharp (\strset{G'}(a)\cap (A\cup B))\geq m-1\geq n$,
    $\sharp (\strset{G'}(a)\cap C)\geq m-1\geq n$, and
    $\sharp (\strset{G'}(a)\cap V)\geq m-1\geq n$.
  \item[\ref{en:p13}] Assume $G, c\models\phi$. Then, there exist
    distinct $x_1,\dots,x_n$ s.t.\ $(c,x_i)\in\iexpr{E}{G}$ and
    $G,x_i\models\phi_1$ for $1\leq i\leq n$. By induction, we consider
    three cases.

    First, if $\iexpr{\phi_1}{G}\cap V=A\cup B$, then
    $x_i,\dots,x_n\in A\cup B$. Therefore, we know $U$ must at least
    contain strings of type 1 or 2. Whenever $U$ contains 1 or 2, we
    verify that
    $\sharp (\strset{G'}(c)\cap (A\cup B))\geq m
    \geq n$.

    Next, if $\iexpr{\phi_1}{G}\cap V=C$, then $x_i,\dots,x_n\in
    C$. Therefore, we know $U$ must at least contain strings of type 5
    or 7. Whenever $U$ contains 5, we verify that
    $\sharp (\strset{G'}(c)\cap C)\geq m \geq
    n$. Otherwise, whenever $U$ contains strings of type 7, and not of
    type 5, we know $n=1$ and clearly
    $\sharp (\strset{G'}(c)\cap C)=1$.

    Next, if $\iexpr{\phi_1}{G}\cap V=V$, then $x_i,\dots,x_n\in
    V$. Therefore, we know $U$ must at least contain strings of type
    1, 2, 5 or 7. All these types have already been handled in the
    previous two cases.

    Finally, the case where $\iexpr{\phi_1}{G}\cap V=\emptyset$ cannot
    occur as $n>0$.

    The sets of strings $U$ above are also exactly the sets used to
    argue the implication from right to left.
    
  \item[\ref{en:p14}] For the direction from left to right, take
    $x\in \iexpr{\phi}{G}\setminus V$. Since $G,x\models\phi$, there
    exists $y_1,\dots,y_n$ s.t.\ $(x,y_i)\in\iexpr{E}{G}$ and
    $G,y_i\models\phi_1$ for $i=1,\dots,n$. However, since $x\not\in V$,
    by Lemma \ref{lem:safety}, all $y_i$ must equal $x$. Hence, $n=1$
    and $(x,x)\in\iexpr{E}{G}$ and $G,x\models\phi_1$. Then again, by
    the same Lemma, $(x,x)\in\iexpr{E}{G'}$, since $G$ and $G'$ have
    the same set of nodes $V$. Moreover, by induction,
    $G',x\models\phi_1$. We conclude $G',x\models\phi$ as desired. The
    direction from right to left is argued symmetrically. \qedhere
  \end{enumerate}
\end{proof}

\subsection{Full disjointness}
\label{sec:fulldisj}
We present here the proof for the primitivity of full disjointness
tests. The general strategy is the same as in 
Section~\ref{sec:fulleq}.

We begin by defining the graphs $G$ and $G'$ formally.
\begin{defi}{$G_\fulldisj(\voc, m)$ }
  Let $m$ be a natural number that is a multiple of $8$. Let
  $A=\{a_1,\dots,a_m\}$, $B=\{b_1,\dots,b_m\}$ and
  $C=\{c_1,\dots,c_m\}$ be three disjoint sets of nodes, disjoint from
  $\voc$. For any $i\leq j$, we write $a_{i\to j}$ to denote the set
  \[\{a_{1+(i-1+l\mod m)}\mid 0\leq l\leq j-i\}\]
  We define $b_{i\to j}$ and $c_{i\to j}$ analogously.

  We define the graph $G_\fulldisj(\voc, m)$ by:
  \begin{align*}
    \iexpr{p}{G}(c_i) &= a_{i\to i+\frac{m}{2}-1} \cup
    b_{i-\frac{m}{8}\to i+\frac{m}{2}-1}  \;\text{and} \\
    \iexpr{q}{G}(c_i) &= a_{i-\frac{m}{2}\to i-1}\cup
    b_{i-\frac{m}{2}\to i+\frac{m}{8}-1} \;\text{for $1\leq i\leq m$}
  \end{align*}
  The $p$ and $q$ relations are visualized in
  Figure~\ref{fig:disjgraph1}.
\end{defi}

To give an example for our notation, suppose $m=8$. Then,
\begin{align*}
  a_{2\to 5} &=\{a_{1+(1+l\mod 8)}\mid 0\leq l\leq 3\}=\{a_2, a_3, a_4, a_5\}\\
  a_{7\to 10}&=\{a_{1+(6+l\mod 8)}\mid 0\leq l\leq 3\}=\{a_7, a_8, a_1,a_2\}, \;\text{and}\\
  a_{-4\to -1} &= \{a_{1+(-5+l\mod 8)}\mid 0\leq l\leq 3\} = \{a_4,a_5,a_6,a_7\}
\end{align*}

\begin{defi}{$G'_\fulldisj(\voc, m)$ }
  We define the graph $G'_\fulldisj(\voc, m)$ on the same nodes as
  $G_\fulldisj(\voc, m)$, with the only difference being the
  relationship of the $p$- and $q$-edges from $C$ to $A$:
  \begin{align*}
    \iexpr{p}{G}(c_i) &= a_{i-\frac{m}{8}\to i+\frac{m}{2}-1} \cup
                        b_{i-\frac{m}{8}\to i+\frac{m}{2}-1}\;\text{and}\\
    \iexpr{q}{G}(c_i) &= a_{i-\frac{m}{2}\to i+\frac{m}{8}-1} \cup
                        b_{i-\frac{m}{2}\to i+\frac{m}{8}-1}\;\text{for $1\leq i\leq m$}
  \end{align*}
  The $p$ and $q$ relations are visualized in
  Figure~\ref{fig:disjgraph1}.
\end{defi}

Important to the intuition behind these graphs is the overlap
generated by the inverse $p$- and $q$-edges. As demonstrated in
Figure~\ref{fig:disjgraph2}, in graph $G=G_\fulldisj(\Sigma,m)$, the
set of $c$ nodes reached from $a$ nodes with inverse $p$ edges is
disjoint from the set of $c$ nodes reached with inverse $q$
edges. This is not the case for $b$ nodes: there, these sets overlap
by precisely one fourth of the $c$ nodes. For graph
$G'=G'_\fulldisj(\Sigma,m)$, the sets of $c$ nodes reachable by
inverse $p$ and $q$ edges overlap for both $a$ and $b$ nodes.

We precisely characterize the behavior of strings on the graphs
$G$ and $G'$ as follows.

\begin{lem}
  \label{lem:disj-pe}
  Let $\voc$ be a vocabulary. Let $m$ be a natural number that is a
  multiple of $8$. Let $G$ be $G_\fulldisj(\voc, m)$ and let $G'$ be
  $G'_\fulldisj(\voc, m)$. The only possibilities for a string $s$
  evaluated on $G$ and $G'$ are the following:
  \begin{enumerate}
  \item
    $\iexpr{s}{G} = \iexpr{p}{G} = \bigcup_{i\in\{1,\dots,m\}}\{c_i\}\times (a_{i\to
      i+\frac{m}{2}-1} \cup b_{i-\frac{m}{8}\to i+\frac{m}{2}-1})$ and \\
    $\iexpr{s}{G'} = \iexpr{p}{G'} = \bigcup_{i\in\{1,\dots,m\}}\{c_i\}\times
    (a_{i-\frac{m}{8}\to i+\frac{m}{2}-1} \cup b_{i-\frac{m}{8}\to
      i+\frac{m}{2}-1})$;
  \item
    $\iexpr{s}{G} = \iexpr{q}{G} = \bigcup_{i\in\{1,\dots,m\}}\{c_i\}\times
    (a_{i-\frac{m}{2}\to i-1}\cup b_{i-\frac{m}{2}\to
      i+\frac{m}{8}-1})$ and \\
    $\iexpr{s}{G'} = \iexpr{q}{G} = \bigcup_{i\in\{1,\dots,m\}}\{c_i\}\times
    (a_{i-\frac{m}{2}\to i+\frac{m}{8}-1} \cup b_{i-\frac{m}{2}\to
      i+\frac{m}{8}-1})$;
  \item
    $\iexpr{s}{G} = \iexpr{p^-}{G} = \bigcup_{i\in\{1,\dots,m\}}
    (\{a_i\}\times c_{i-\frac{m}{2}+1\to i})\cup (\{b_i\}\times
    c_{i-\frac{m}{2}+1\to
      i+\frac{m}{8}})$ and \\
    $\iexpr{s}{G'} = \iexpr{p^-}{G'} = \bigcup_{i\in\{1,\dots,m\}}
    (\{a_i\}\times c_{i-\frac{m}{2}+1\to i+\frac{m}{8}})\cup
    (\{b_i\}\times c_{i-\frac{m}{2}+1\to i+\frac{m}{8}})$ ;
  \item
    $\iexpr{s}{G} = \iexpr{q^-}{G} = \bigcup_{i\in\{1,\dots,m\}}
    (\{a_i\}\times c_{i+1\to i+\frac{m}{2}})\cup (\{b_i\}\times
    c_{i-\frac{m}{8}+1\to
      i+\frac{m}{2}})$ and \\
    $\iexpr{s}{G'} = \iexpr{q^-}{G'} = \bigcup_{i\in\{1,\dots,m\}}
    (\{a_i\}\times c_{i-\frac{m}{8}+1\to i+\frac{m}{2}})\cup
    (\{b_i\}\times c_{i-\frac{m}{8}+1\to i+\frac{m}{2}})$;
  \item $\iexpr{s}{G} = \iexpr{s}{G'} = C\times C$;
  \item $\iexpr{s}{G} = \iexpr{s}{G'} = (A\cup B)\times (A\cup B)$;
  \item $\iexpr{s}{G} = \iexpr{s}{G'} = C\times (A\cup B)$;
  \item $\iexpr{s}{G} = \iexpr{s}{G'} = (A\cup B)\times C$;
  \item $\iexpr{s}{G} = \iexpr{s}{G'} = \id$; or
  \item $\iexpr{s}{G} = \iexpr{s}{G'} = \emptyset$
  \end{enumerate}
  The first four types of strings are visualized in
  Figure~\ref{fig:disjgraph1} and Figure~\ref{fig:disjgraph2}.
\end{lem}

\begin{proof}
  The proof is performed as in the proof of Lemma \ref{lem:eq-pe}. We
  now have the following table:
  
  \begin{center}
  \begin{tabular}[b]{lcc}
    \toprule
    $s$ & $\iexpr{s}{G}$ & $\iexpr{s}{G'}$ \\
    \midrule
    $\mathbf{id}$ & $\id$ & $\id$ \\ 
    $\mathbf{p}$ & type 1 & type 1\\
    $\mathbf{q}$ & type 2 & type 2 \\
    $\mathbf{p^-}$ & type 3 & type 3 \\
    $\mathbf{q^-}$ & type 4 & type 4 \\
    $\mathbf{p\comp p^-}$ & $C\times C$ & $C\times C$ \\
    $p\comp q^-$ & $C\times C$ & $C\times C$ \\
    $q\comp p^-$ & $C\times C$ & $C\times C$ \\
    $q\comp q^-$ & $C\times C$ & $C\times C$ \\
    $\mathbf{p^-\comp p}$ & $(A\cup B)\times (A\cup B)$ & $(A\cup B)\times (A\cup B)$ \\
    $p^-\comp q$ & $(A\cup B)\times (A\cup B)$ & $(A\cup B)\times (A\cup B)$ \\
    $q^-\comp p$ & $(A\cup B)\times (A\cup B)$ & $(A\cup B)\times (A\cup B)$ \\
    $q^-\comp q$ & $(A\cup B)\times (A\cup B)$ & $(A\cup B)\times (A\cup B)$ \\
    $\mathbf{p\comp p^-\comp p}$ & $C\times (A\cup B)$ & $C\times (A\cup B)$ \\
    $p\comp p^-\comp q$ & $C\times (A\cup B)$ & $C\times (A\cup B)$ \\
    $\mathbf{p^-\comp p\comp p^-}$ & $(A\cup B)\times C$ & $(A\cup B)\times C$ \\
    $p^-\comp p\comp q^-$ & $(A\cup B)\times C$ & $(A\cup B)\times C$ \\
    $p\comp p^-\comp p\comp p^-$ & $C\times C$ & $C\times C$ \\
    $p\comp p^-\comp p\comp q^-$ & $C\times C$ & $C\times C$ \\
    $p^-\comp p\comp p^-\comp p$ & $(A\cup B)\times (A\cup B)$ & $(A\cup B)\times (A\cup B)$ \\
    $p^-\comp p\comp p^-\comp q$ & $(A\cup B)\times (A\cup B)$ & $(A\cup B)\times (A\cup B)$ \\
    \bottomrule
  \end{tabular}\\\qedhere
\end{center}
\end{proof}

\newcommand{\drawcircle}[1]{
  \draw (0,1) node[above] {#1};
  \begin{scope}[shift={(-0.25,0.97)}]
    \draw [gray,rotate=-73,-Stealth] (0,0) -- (0,0.1);
  \end{scope}
  \draw [gray,dashed] (0,1) arc [start angle=450, end angle=105, delta angle=15, radius=1cm ];}

\begin{figure}
  $\iexpr{p}{G}=\bigcup_{i\in\{1,\dots,m\}}\{c_i\}\times (a_{i\to
    i+\frac{m}{2}-1} \cup b_{i-\frac{m}{8}\to i+\frac{m}{2}-1})$
  
  \begin{minipage}[c]{0.3\linewidth}
    \begin{tikzpicture}
      \drawcircle{$c_i$}
      \draw[-Circle] (0,1.06);
    \end{tikzpicture}
  \end{minipage}
  \begin{minipage}[c]{0.3\linewidth}
    \begin{tikzpicture}
      \drawcircle{$a_i$}
      \begin{scope}[shift={(-0.07,0)}]
        \draw[{Circle-Circle[open]}] (0,1) arc [start angle=93.5, end angle=-93.5, radius=1cm ];
      \end{scope}
    \end{tikzpicture}
  \end{minipage}
  \begin{minipage}[c]{0.3\linewidth}
    \begin{tikzpicture}
      \drawcircle{$b_i$}
      \begin{scope}[rotate=48.5]
        \draw[{Circle-Circle[open]}] (0,1) arc [start angle=90, end angle=-142, radius=1cm ];
      \end{scope}
    \end{tikzpicture}
  \end{minipage}

  $\iexpr{q}{G}=\bigcup_{i\in\{1,\dots,m\}}\{c_i\}\times
  (a_{i-\frac{m}{2}\to i-1}\cup b_{i-\frac{m}{2}\to i+\frac{m}{8}-1})$
  
  \begin{minipage}[c]{0.3\linewidth}
    \begin{tikzpicture}
      \drawcircle{$c_i$}
      \draw[-Circle] (0,1.06);
    \end{tikzpicture}
  \end{minipage}
  \begin{minipage}[c]{0.3\linewidth}
    \begin{tikzpicture}
      \drawcircle{$a_i$}
      \begin{scope}[shift={(0.07,0)}]
        \draw[{Circle[open]}-Circle] (0,1) arc [start angle=86.5, end angle=273.5, radius=1cm ];
      \end{scope}
    \end{tikzpicture}
  \end{minipage}
  \begin{minipage}[c]{0.3\linewidth}
    \begin{tikzpicture}
      \drawcircle{$b_i$}
      \begin{scope}[rotate=183.5]
        \draw[{Circle-Circle[open]}] (0,1) arc [start angle=90, end angle=-142, radius=1cm ];
      \end{scope}
    \end{tikzpicture}
  \end{minipage}

  $\iexpr{p}{G'}=\bigcup_{i\in\{1,\dots,m\}}\{c_i\}\times
  (a_{i-\frac{m}{8}\to i+\frac{m}{2}-1} \cup b_{i-\frac{m}{8}\to
    i+\frac{m}{2}-1})$

  \begin{minipage}[c]{0.3\linewidth}
    \begin{tikzpicture}
      \drawcircle{$c_i$}
      \draw[-Circle] (0,1.06);
    \end{tikzpicture}
  \end{minipage}
  \begin{minipage}[c]{0.3\linewidth}
    \begin{tikzpicture}
      \drawcircle{$a_i$}
      \begin{scope}[rotate=48.5]
        \draw[{Circle-Circle[open]}] (0,1) arc [start angle=90, end angle=-142, radius=1cm ];
      \end{scope}
    \end{tikzpicture}
  \end{minipage}
  \begin{minipage}[c]{0.3\linewidth}
    \begin{tikzpicture}
      \drawcircle{$b_i$}
      \begin{scope}[rotate=48.5]
        \draw[{Circle-Circle[open]}] (0,1) arc [start angle=90, end angle=-142, radius=1cm ];
      \end{scope}
    \end{tikzpicture}
  \end{minipage}
  
  $\iexpr{q}{G'}=\bigcup_{i\in\{1,\dots,m\}}\{c_i\}\times
  (a_{i-\frac{m}{2}\to i+\frac{m}{8}-1} \cup b_{i-\frac{m}{2}\to
    i+\frac{m}{8}-1})$

  \begin{minipage}[c]{0.3\linewidth}
    \begin{tikzpicture}
      \drawcircle{$c_i$}
      \draw[-Circle] (0,1.06);
    \end{tikzpicture}
  \end{minipage}
  \begin{minipage}[c]{0.3\linewidth}
    \begin{tikzpicture}
      \drawcircle{$a_i$}
      \begin{scope}[rotate=183.5]
        \draw[{Circle-Circle[open]}] (0,1) arc [start angle=90, end angle=-142, radius=1cm ];
      \end{scope}
    \end{tikzpicture}
  \end{minipage}
  \begin{minipage}[c]{0.3\linewidth}
    \begin{tikzpicture}
      \drawcircle{$b_i$}
      \begin{scope}[rotate=183.5]
        \draw[{Circle-Circle[open]}] (0,1) arc [start angle=90, end angle=-142, radius=1cm ];
      \end{scope}
    \end{tikzpicture}
  \end{minipage}
  \caption{Illustration of the $p$ and $q$ relations in graphs $G=G_\fulldisj(\voc,m)$ and $G'=G'_\fulldisj(\voc,m)$}
  \label{fig:disjgraph1}
\end{figure}
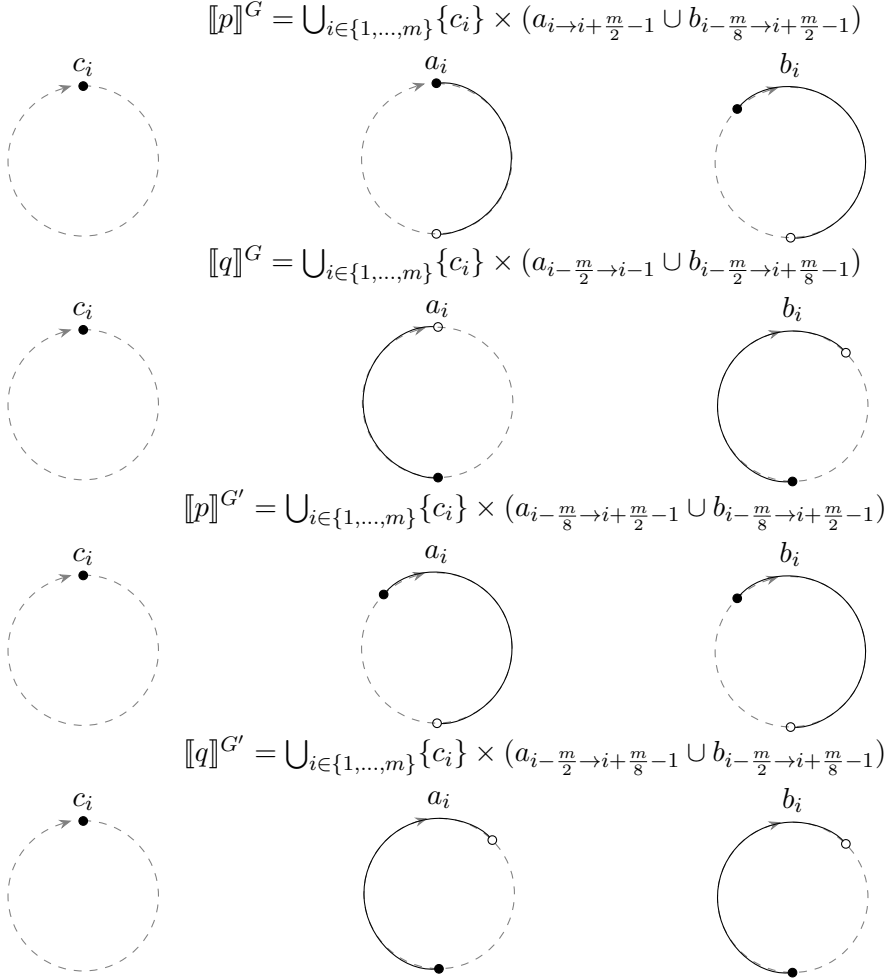

\begin{figure}
  $\iexpr{p^-}{G} = \bigcup_{i\in\{1,\dots,m\}} (\{a_i\}\times
  c_{i-\frac{m}{2}+1\to i})\cup (\{b_i\}\times c_{i-\frac{m}{2}+1\to
    i+\frac{m}{8}})$
  
  \begin{minipage}[c]{0.24\linewidth}
    \begin{tikzpicture}
      \drawcircle{$a_i$}
      \draw[-Circle] (0,1.06);
    \end{tikzpicture}
  \end{minipage}
  \begin{minipage}[c]{0.24\linewidth}
    \begin{tikzpicture}
      \drawcircle{$c_i$}
      \begin{scope}[shift={(0.07,0)}]
        \draw[Circle-{Circle[open]}] (0,1) arc [start angle=86.5, end angle=273.5, radius=1cm ];
      \end{scope}
    \end{tikzpicture}
  \end{minipage}
  \begin{minipage}[c]{0.24\linewidth}
    \begin{tikzpicture}
      \drawcircle{$b_i$}
      \draw[-Circle] (0,1.06);
    \end{tikzpicture}
  \end{minipage}
  \begin{minipage}[c]{0.24\linewidth}
    \begin{tikzpicture}
      \drawcircle{$c_i$}
      \begin{scope}[rotate=183.5]
        \draw[{Circle[open]}-Circle] (0,1) arc [start angle=90, end angle=-142, radius=1cm ];
      \end{scope}
    \end{tikzpicture}
  \end{minipage}

  $\iexpr{q^-}{G} = \bigcup_{i\in\{1,\dots,m\}} (\{a_i\}\times
  c_{i+1\to i+\frac{m}{2}})\cup (\{b_i\}\times c_{i-\frac{m}{8}+1\to
    i+\frac{m}{2}})$
  
  \begin{minipage}[c]{0.24\linewidth}
    \begin{tikzpicture}
      \drawcircle{$a_i$}
      \draw[-Circle] (0,1.06);
    \end{tikzpicture}
  \end{minipage}
  \begin{minipage}[c]{0.24\linewidth}
    \begin{tikzpicture}
      \drawcircle{$c_i$}
      \begin{scope}[shift={(-0.07,0)}]
        \draw[{Circle[open]}-Circle] (0,1) arc [start angle=93.5, end angle=-93.5, radius=1cm ];
      \end{scope}
    \end{tikzpicture}
  \end{minipage}
  \begin{minipage}[c]{0.24\linewidth}
    \begin{tikzpicture}
      \drawcircle{$b_i$}
      \draw[-Circle] (0,1.06);
    \end{tikzpicture}
  \end{minipage}
  \begin{minipage}[c]{0.24\linewidth}
    \begin{tikzpicture}
      \drawcircle{$c_i$}
      \begin{scope}[rotate=48.5]
        \draw[{Circle[open]}-Circle] (0,1) arc [start angle=90, end angle=-142, radius=1cm ];
      \end{scope}
    \end{tikzpicture}
  \end{minipage}

  $\iexpr{p^-}{G'} = \bigcup_{i\in\{1,\dots,m\}}
  (\{a_i\}\times c_{i-\frac{m}{2}+1\to i+\frac{m}{8}})\cup
  (\{b_i\}\times c_{i-\frac{m}{2}+1\to i+\frac{m}{8}})$
  
  \begin{minipage}[c]{0.24\linewidth}
    \begin{tikzpicture}
      \drawcircle{$a_i$}
      \draw[-Circle] (0,1.06);
    \end{tikzpicture}
  \end{minipage}
  \begin{minipage}[c]{0.24\linewidth}
    \begin{tikzpicture}
      \drawcircle{$c_i$}
      \begin{scope}[rotate=183.5]
        \draw[{Circle[open]}-Circle] (0,1) arc [start angle=90, end angle=-142, radius=1cm ];
      \end{scope}
    \end{tikzpicture}
  \end{minipage}
  \begin{minipage}[c]{0.24\linewidth}
    \begin{tikzpicture}
      \drawcircle{$b_i$}
      \draw[-Circle] (0,1.06);
    \end{tikzpicture}
  \end{minipage}
  \begin{minipage}[c]{0.24\linewidth}
    \begin{tikzpicture}
      \drawcircle{$c_i$}
      \begin{scope}[rotate=183.5]
        \draw[{Circle[open]}-Circle] (0,1) arc [start angle=90, end angle=-142, radius=1cm ];
      \end{scope}
    \end{tikzpicture}
  \end{minipage}

  $\iexpr{q^-}{G'} = \bigcup_{i\in\{1,\dots,m\}}
  (\{a_i\}\times c_{i-\frac{m}{8}+1\to i+\frac{m}{2}})\cup
  (\{b_i\}\times c_{i-\frac{m}{8}+1\to i+\frac{m}{2}})$
  
  \begin{minipage}[c]{0.24\linewidth}
    \begin{tikzpicture}
      \drawcircle{$a_i$}
      \draw[-Circle] (0,1.06);
    \end{tikzpicture}
  \end{minipage}
  \begin{minipage}[c]{0.24\linewidth}
    \begin{tikzpicture}
      \drawcircle{$c_i$}
      \begin{scope}[rotate=48.5]
        \draw[{Circle[open]}-Circle] (0,1) arc [start angle=90, end angle=-142, radius=1cm ];
      \end{scope}
    \end{tikzpicture}
  \end{minipage}
  \begin{minipage}[c]{0.24\linewidth}
    \begin{tikzpicture}
      \drawcircle{$b_i$}
      \draw[-Circle] (0,1.06);
    \end{tikzpicture}
  \end{minipage}
  \begin{minipage}[c]{0.24\linewidth}
    \begin{tikzpicture}
      \drawcircle{$c_i$}
      \begin{scope}[rotate=48.5]
        \draw[{Circle[open]}-Circle] (0,1) arc [start angle=90, end angle=-142, radius=1cm ];
      \end{scope}
    \end{tikzpicture}
  \end{minipage}
  \caption{Illustration of the $p^-$ and $q^-$ relations in graphs
    $G=G_\fulldisj(\voc,m)$ and $G'=G'_\fulldisj(\voc,m)$}
  \label{fig:disjgraph2}
\end{figure}
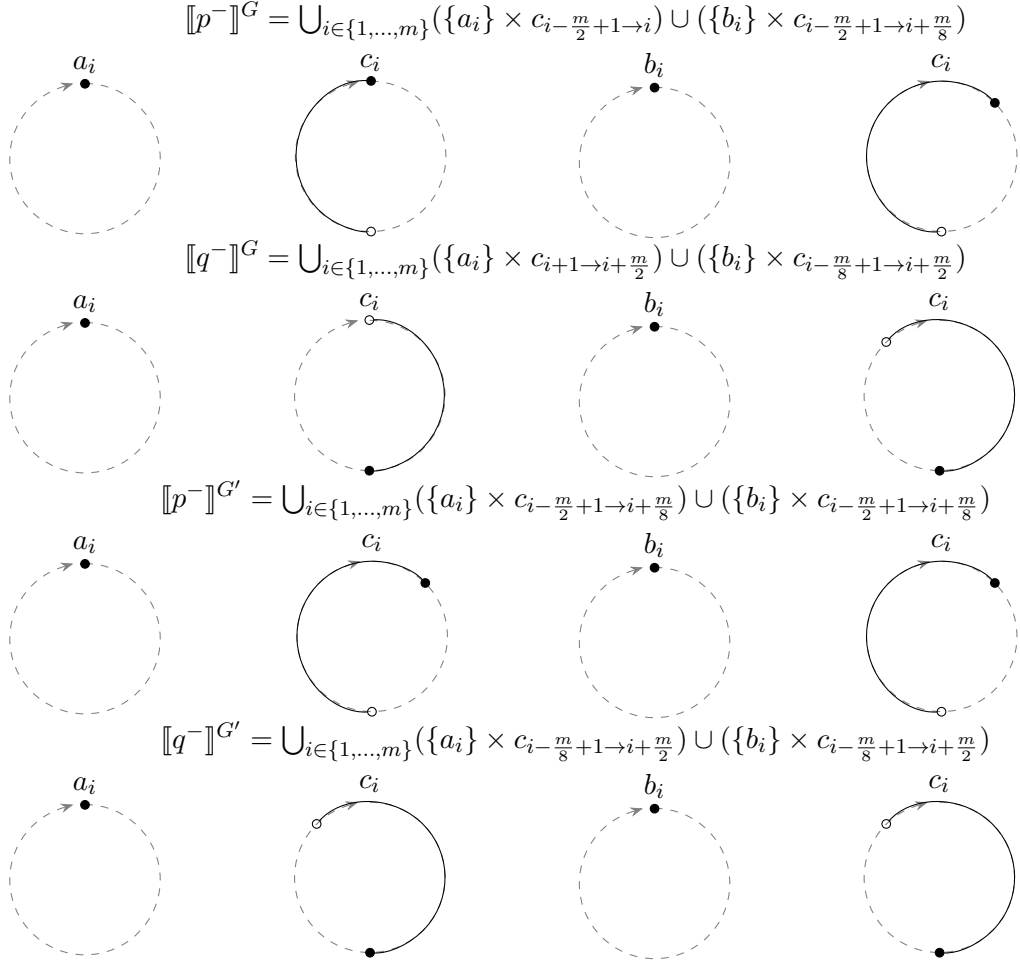

We are ready to present our key Proposition.
\begin{prop} \label{prop:fulldisj}
  Let $\voc$ be a vocabulary. Let $m$ be a natural number and a
  multiple of 8. Let $V=A\cup B\cup C$ be the common set of nodes of
  the graphs $G = G_\fulldisj(\voc, m)$ and
  $G' = G'_\fulldisj(\voc, m)$. For all shapes $\phi$ over $\voc$
  counting to at most $\frac{m}{2}$, we have
  $\iexpr{\phi}{G} = \iexpr{\phi}{G'}$. Moreover,
  \begin{itemize}
  \item $\iexpr{\phi}{G} \cap V = A\cup B$, or
  \item $\iexpr{\phi}{G} \cap V = C$, or
  \item $\iexpr{\phi}{G} \cap V = V$, or
  \item $\iexpr{\phi}{G} \cap V = \emptyset$.
  \end{itemize}
\end{prop}

\begin{proof}
  By induction on the structure of $\phi$. For the base cases, if
  $\phi$ is $\top$ then $\iexpr{\top}{G}=\iexpr{\top}{G'}=N$ and
  $N\cap V = V$. If $\phi$ is $\{c\}$, then
  $\iexpr{\{c\}}{G}=\iexpr{\{c\}}{G'} = \{c\}$ and
  $\{c\}\cap V=\emptyset$ since $c\in\voc$ and $V\cap\voc=\emptyset$.

  If $\phi$ is $\closed(Q)$, we consider the possibilities for $Q$. If
  $Q$ does not contain both $p$ and $q$, then clearly
  $\iexpr{\phi}{G}\cap V=\iexpr{\phi}{G'}\cap V=A\cup B$. Otherwise,
  $\iexpr{\phi}{G}=\iexpr{\phi}{G'}=N$.

  Before considering the remaining cases, we observe the following
  symmetries:
  \begin{itemize}
  \item In both $G$ and $G'$, all elements of $A$ are symmetrical, as
    are all elements of $B$, and all elements of $C$. Indeed, for any
    $i\in\{1,\dots,m\}$, the function that maps $x_i$ to
    $x_{1+i\mod m}$ where $x_i$ is $a_i$, $b_i$ or $c_i$, is clearly an
    automorphism of $G$ and also of $G'$.
  \item Furthermore, in $G'$, any $a_i$ and $b_j$ are
    symmetrical. Indeed, the function that swaps every $a_i$ with
    $b_i$ is an automorphism of $G'$. (We already know that $b_i$ and
    $b_j$ are symmetrical by the above.)
  \end{itemize}
  Therefore, we are only left to show:
  \begin{enumerate}[(i)]
  \item \label{en:p21} For any $a\in A$ and $b\in B$, we have
    $G,a\models\phi\iff G,b\models \phi$,
  \item \label{en:p22} For any $a\in A$, we have
    $G,a\models\phi\iff G',a\models\phi$, and
  \item \label{en:p23} For any $c\in C$, we have
    $G,c\models\phi\iff G',c\models\phi$.
  \item \label{en:p24} For any $x\not\in V$, we have
    $G,x\models\phi\iff G',x\models\phi$.
  \end{enumerate}
  Note that then also for any $b\in B$, we have
  $G,b\models\phi\iff G',b\models\phi$ because for any $a\in A$ and
  $b\in B$, we have
  $G,b\models\phi\stackrel{\rm\ref{en:p11}}{\iff} G, a\models \phi
  \stackrel{\rm\ref{en:p12}}{\iff} G',a\models\phi
  \stackrel{\mathit{symmetry}}{\iff} G',b\models\phi$.
  
  Consider the case where $\phi$ is $\disj(E,r)$. We verify
  \ref{en:p21}, \ref{en:p22}, \ref{en:p23}, and \ref{en:p24}.  First,
  to see that \ref{en:p21} and \ref{en:p22} hold, we observe that
  $\iexpr{r}{G}(a)=\iexpr{r}{G}(b)= \iexpr{r}{G'}(a)=
  \iexpr{r}{G'}(b)= \emptyset$. Therefore, $G,a\models\phi$,
  $G,b\models\phi$, $G',a\models\phi$, and $G',b\models\phi$ always
  hold, showing \ref{en:p21} and \ref{en:p22}.

  Next, to show \ref{en:p23} where $r=p$, assume
  $G,c\models\disj(E,p)$.  By Lemma \ref{lem:strings} we know there
  is a set of strings $U$ equivalent to $E$ in both $G$ and $G'$. By
  Lemma \ref{lem:disj-pe} there are only 10 types of strings. We
  observe from Lemma \ref{lem:disj-pe} that for every $U$,
  $\strset{G}(c)$ is disjoint from
  $\iexpr{p}{G}(c)$ whenever $U$ does not contain strings of type 1,
  2, or 7. These are also exactly the $U$ s.t.\
  $\strset{G'}(c)$ is disjoint from
  $\iexpr{p}{G'}(c)$.

  Next, to show \ref{en:p23} where $r=q$, assume
  $G,c\models\disj(E,q)$. We observe from Lemma \ref{lem:disj-pe}
  that for every $U$, $\strset{G}(c)$ is disjoint
  from $\iexpr{q}{G}(c)$ whenever $U$ does not contain strings of
  type 1, 2, or 7. These are also exactly the $U$ s.t.\
  $\strset{G'}(c)$ is disjoint from
  $\iexpr{q}{G'}(c)$.

  For every other property name $r$,
  $\iexpr{r}{G}(a)=\iexpr{r}{G}(b)=\emptyset$. Therefore,
  $G,c\models\phi$ and $G',c\models\phi$ always hold.

  Finally, we show \ref{en:p24} by observing that for any
  $x\in N\setminus V$,
  $\iexpr{r}{G}(x)=\iexpr{r}{G'}(x)=\emptyset$. Therefore,
  $G,x\models\phi$ and $G',x\models\phi$ always hold.
  
  Next, consider the case where $\phi$ is $eq(E_1,E_2)$. We again
  verify \ref{en:p21}, \ref{en:p22}, \ref{en:p23}, and
  \ref{en:p24}.

  We show \ref{en:p21} by using a canonical labeling argument. For any
  two sets $U_1$ and $U_2$ of types, we call $U_1$ and $U_2$
  equivalent in $a\in A$ if
  $\bigcup_{s\in U_1}\iexpr{s}{G}(a) = \bigcup_{s\in
    U_2}\iexpr{s}{G}(a)$. Similarly, we define when $U_1$ and $U_2$
  are equivalent in $b$ or $c$.

  We can canonically label the equivalence
  classes in $a$ as follows. Let $U=\{u_1,\dots,u_l\}$ and
  $u_1 < \dots < u_l$ with each $u_i\in \{1,\dots, 10\}$ a type.

  There are only six unique singleton sets namely $\{1\}$, $\{3\}$,
  $\{4\}$, $\{6\}$, $\{8\}$, and $\{9\}$. Replace each $u_i$ by their
  singleton representative $\overline{u_i}$. In
  $\{\overline{u_1}, \dots, \overline{u_l}\}$, reorder and remove
  duplicates to obtain an equivalent set $\{u'_1,\dots,u'_{l'}\}$.

  If $l'=1$, we are done. Otherwise, we enumerate all nonequivalent
  2-element sets that are not equivalent to a singleton: there are
  again six of those, namely $\{3,6\}$, $\{3,9\}$, $\{4,6\}$,
  $\{4,9\}$, $\{6,8\}$, and $\{8,9\}$.

  Replace $u'_1$ and $u'_2$ by either $\{u''\}$ in case
  $\{u'_1, u'_2\}$ is equivalent to a singleton; otherwise replace
  $u'_1$ and $u'_2$ by their equivalent 2-element set
  $\{u''_1, u''_2\}$. If $l'=2$, we are again done.

  We can repeat this process. However, it turns out that there are no
  3-element sets that are not equivalent to a singleton or a 2-element
  set. Hence, there are only 12 representatives. The enumeration
  process is shown in Table \ref{table:canonical-ab}, giving the
  representative for equivalence in $a$ as well as for equivalence in
  $b$. Crucially, in filling the table, we observe every set $U$ has
  the same representative for equivalence in $a$ as for equivalence in
  $b$.

  Next, \ref{en:p22} is shown in an analogous manner, where the
  enumeration process is again shown in Table
  \ref{table:canonical-ab}.

  Next, \ref{en:p23} is again shown with an analogous manner, where
  the enumeration process is shown in Table
  \ref{table:canonical-c}.

  To show \ref{en:p24}, assume $x\not\in V$. If both $E_1$ and $E_2$
  are safe, then by Lemma \ref{lem:safety}
  $\iexpr{E_1}{G}(x)=\iexpr{E_2}{G'}(x)=\emptyset$. Thus,
  $G,x\models\phi$ and $G',x\models\phi$. If both $E_1$ and $E_2$ are
  unsafe, then by Lemma \ref{lem:safety}
  $\iexpr{E_1}{G}(x)=\iexpr{E_2}{G'}(x)=\{x\}$. Thus, $G,x\models\phi$
  and $G',x\models\phi$. However, whenever only one of $E_1$ and $E_2$
  is safe, clearly $G,x\not\models\phi$ and $G',x\not\models\phi$.

  The cases where $\phi$ is $\phi_1\land\phi_2$, $\phi_1\lor\phi_2$ or
  $\neg\phi_1$ are handled by induction in a straightforward manner.

  Lastly, we consider the case where $\phi$ is $\geqn{n}{E}{\psi}$.
  \begin{enumerate}
  \item[\ref{en:p21}] Assume $G, a\models\phi$. Then, there exist distinct
    $x_1,\dots,x_n$ s.t.\ $(a,x_i)\in\iexpr{E}{G}$ and
    $G,x_i\models\psi$ for $1\leq i\leq n$. Again, by Lemma
    \ref{lem:strings} we know there is a set of strings $U$
    equivalent to $E$ in both $G$ and $G'$. By Lemma
    \ref{lem:disj-pe} there are only 10 types of strings.  By
    induction, we consider three cases.

    First, if $\iexpr{\psi}{G}\cap V=A\cup B$, then
    $x_i,\dots,x_n\in A\cup B$. Therefore, we know $U$ must at least
    contain strings of type 6 or 9. Suppose $U$ contains strings of
    type 6. Then, we verify that
    $\sharp (\strset{G}(b)\cap (A\cup B))\geq m \geq
    n$. Otherwise, whenever $U$ contains strings of type 9, and not of
    type 6, we know $n=1$ and clearly
    $\sharp (\strset{G}(b)\cap (A\cup B))=1$.

    Next, if $\iexpr{\psi}{G}\cap V=C$, then $x_i,\dots,x_n\in
    C$. Therefore, we know $U$ must at least contain strings of type
    3, 4 or 8. Whenever $U$ contains 3, 4 or 8, we verify that
    $\sharp (\strset{G}(b)\cap C)\geq \frac{m}{2} \geq n$.

    Next, if $\iexpr{\psi}{G}\cap V=V$, then $x_i,\dots,x_n\in
    V$. Therefore, we know $U$ must at least contain strings of type
    3, 4, 6, 8 or 9. All these types have already been handled in the
    previous two cases.

    Finally, the case where $\iexpr{\psi}{G}\cap V=\emptyset$ cannot
    occur as $n>0$.

    The sets of strings $U$ above are also exactly the sets used to
    argue the implication from right to left.
  \item[\ref{en:p22}] For every case of $U$ above, for every inductive case of
    $\psi$, we can also verify that
    $\sharp (\strset{G'}(a)\cap (A\cup B))\geq m \geq n$,
    $\sharp (\strset{G'}(a)\cap C)\geq \frac{m}{2} \geq n$, and
    $\sharp (\strset{G'}(a)\cap V)\geq \frac{m}{2} \geq n$.
  \item[\ref{en:p23}] Assume $G, c\models\phi$. Then, there exist distinct
    $x_1,\dots,x_n$ s.t.\ $(c,x_i)\in\iexpr{E}{G}$ and
    $G,x_i\models\psi$ for $1\leq i\leq n$. By induction, we consider
    three cases.

    First, if $\iexpr{\psi}{G}\cap V=A\cup B$, then
    $x_i,\dots,x_n\in A\cup B$. Therefore, we know $U$ must at least
    contain strings of type 1, 2 or 7. Whenever $U$ contains 1, 2 or
    7, we verify that
    $\sharp (\strset{G'}(c)\cap (A\cup B))\geq \frac{m}{2} \geq n$.

    Next, if $\iexpr{\psi}{G}\cap V=C$, then $x_i,\dots,x_n\in
    C$. Therefore, we know $U$ must at least contain strings of type 5
    or 9. Whenever $U$ contains 5, we verify that
    $\sharp (\strset{G'}(c)\cap C)\geq m \geq
    n$. Otherwise, whenever $U$ contains strings of type 9, and not of
    type 5, we know $n=1$ and clearly
    $\sharp (\strset{G'}(c)\cap C)=1$.

    Next, if $\iexpr{\psi}{G}\cap V=V$, then $x_i,\dots,x_n\in
    V$. Therefore, we know $U$ must at least contain strings of type
    1, 2, 5, 7 or 9. All these types have already been handled in the
    previous two cases.

    Finally, the case where $\iexpr{\psi}{G}\cap V=\emptyset$ cannot
    occur as $n>0$.

    The sets of strings $U$ above are also exactly the sets used to
    argue the implication from right to left.
  \item[\ref{en:p24}] For the direction from left to right, take
    $x\in \iexpr{\phi}{G}\setminus V$. Since $G,x\models\phi$, there
    exists $y_1,\dots,y_n$ s.t.\ $(x,y_i)\in\iexpr{E}{G}$ and
    $G,y_i\models\psi$ for $i=1,\dots,n$. However, since $x\not\in V$,
    by Lemma \ref{lem:safety}, all $y_i$ must equal $x$. Hence, $n=1$
    and $(x,x)\in\iexpr{E}{G}$ and $G,x\models\psi$. Then again, by
    the same Lemma, $(x,x)\in\iexpr{E}{G'}$, since $G$ and $G'$ have
    the same set of nodes $V$/ Moreover, by induction,
    $G',x\models\psi$. We conclude $G',x\models\phi$ as desired. The
    direction from right to left is argued symmetrically. \qedhere
  \end{enumerate}
\end{proof}

\begin{table}
  \caption{Sets of types starting from $a_i$, $b_i$ in $G$ and $a_i$ in $G'$.}
  \label{table:canonical-ab}
  \begin{tabular}{lccc}
    \toprule
    $U$ & $\iexpr{E}{G}(a_i)$ & $\iexpr{E}{G}(b_i)$ & $\iexpr{E}{G'}(a_i)$ \\
    \midrule
  $\mathbf{\{1\}}$ & $\emptyset$ & $\emptyset$ & $\emptyset$ \\
  $\{2\}$ & $\emptyset$ & $\emptyset$ & $\emptyset$ \\
  $\mathbf{\{3\}}$ & $c_{i-\frac{k}{2}+1\to i}$ & $c_{i-\frac{k}{2}+1\to i+\frac{k}{8}}$ & $c_{i-\frac{k}{2}+1\to i+\frac{k}{8}}$ \\
  $\mathbf{\{4\}}$ & $c_{i+1\to i+\frac{k}{2}}$ & $c_{i-\frac{k}{8}+1\to i+\frac{k}{2}}$ & $c_{i-\frac{k}{8}+1\to i+\frac{k}{2}}$ \\
  $\{5\}$ & $\emptyset$ & $\emptyset$ & $\emptyset$ \\
  $\mathbf{\{6\}}$ & $A\cup B$ & $A\cup B$ & $A\cup B$ \\
  $\{7\}$ & $\emptyset$ & $\emptyset$ & $\emptyset$ \\
  $\mathbf{\{8\}}$ & $C$ & $C$ & $C$ \\
  $\mathbf{\{9\}}$ & $\{a_i\}$ & $\{b_i\}$ & $\{a_i\}$ \\
  $\{10\}$ & $\emptyset$ & $\emptyset$ & $\emptyset$ \\
  $\{3,4\}$ & $C$ & $C$ & $C$ \\
  $\mathbf{\{3,6\}}$ & $A\cup B\cup c_{i-\frac{k}{2}+1\to i}$ & $A\cup B\cup c_{i-\frac{k}{2}+1\to i+\frac{k}{8}}$ & $A\cup B\cup c_{i-\frac{k}{2}+1\to i+\frac{k}{8}}$ \\
  $\{3,8\}$ & $C$ & $C$ & \\
  $\mathbf{\{3,9\}}$ & $c_{i+1\to i+\frac{k}{2}} \cup \{a_i\}$ & $c_{i-\frac{k}{8}+1\to i+\frac{k}{2}} \cup \{b_i\}$ & $c_{i-\frac{k}{8}+1\to i+\frac{k}{2}} \cup \{a_i\}$\\
  $\mathbf{\{4,6\}}$ & $A\cup B\cup c_{i+1\to i+\frac{k}{2}}$ & $A\cup B\cup c_{i-\frac{k}{8}+1\to i+\frac{k}{2}}$ & $A\cup B\cup c_{i-\frac{k}{8}+1\to i+\frac{k}{2}}$ \\
  $\{4,8\}$ & $C$ & $C$ & \\
  $\mathbf{\{4,9\}}$ & $c_{i+1\to i+\frac{k}{2}}\cup \{a_i\}$ & $c_{i-\frac{k}{8}+1\to i+\frac{k}{2}}\cup \{b_i\}$ & $c_{i-\frac{k}{8}+1\to i+\frac{k}{2}}\cup \{a_i\}$ \\
  $\mathbf{\{6,8\}}$ & $V$ & $V$ & $V$ \\
  $\{6,9\}$ & $A\cup B$ & $A\cup B$ & $A\cup B$\\
  $\mathbf{\{8,9\}}$ & $C\cup \{a_i\}$ & $C\cup \{b_i\}$ & $C\cup \{a_i\}$ \\
  $\{3,6,4\}$ & $V$ & $V$ & $V$ \\
  $\{3,6,8\}$ & $V$ & $V$ & $V$ \\
  $\{3,6,9\}$ & $A\cup B\cup c_{i-\frac{k}{2}+1\to i}$ & $A\cup B\cup c_{i-\frac{k}{2}+1\to i+\frac{k}{8}}$ & $A\cup B\cup c_{i-\frac{k}{2}+1\to i+\frac{k}{8}}$ \\
  $\{3,9,4\}$ & $C\cup \{a_i\}$ & $C\cup \{b_i\}$ & $C\cup \{a_i\}$ \\
  $\{4,9,6\}$ & $A\cup B\cup c_{i+1\to i+\frac{k}{2}}$ & $A\cup B\cup c_{i-\frac{k}{8}+1\to i+\frac{k}{2}}$ & $A\cup B\cup c_{i-\frac{k}{8}+1\to i+\frac{k}{2}}$ \\
  $\{4,9,8\}$ & $C$ & $C$ & $C$ \\
  $\{8,9,3\}$ & $C\cup \{a_i\}$ & $C\cup \{b_i\}$ & $C\cup \{a_i\}$ \\
    $\{4,6,3\}$ & $V$ & $V$ & $V$ \\
    \bottomrule
  \end{tabular}
\end{table}

\begin{table}
  \caption{Sets of types starting from $c_i$ in $G$ and in $G'$.}
  \label{table:canonical-c}
  \begin{tabular}{lcc}
    \toprule
    $U$ & $\iexpr{E}{G}(c_i)$ & $\iexpr{E}{G'}(c_i)$ \\
    \midrule
    $\mathbf{\{1\}}$ &
  $a_{i\to i+\frac{m}{2}-1}\cup b_{i-\frac{m}{8}\to i+\frac{m}{2}-1}$
  &
  $a_{i-\frac{m}{8}\to i+\frac{m}{2}-1}\cup b_{i-\frac{m}{8}\to
    i+\frac{m}{2}-1}$ \\
  $\mathbf{\{2\}}$ & $a_{i-\frac{m}{2}\to i-1}\cup b_{i-\frac{m}{2}\to i+\frac{m}{8}-1}$ &
  $a_{i-\frac{m}{2}\to i+\frac{m}{8}-1}\cup b_{i-\frac{m}{2}\to i+\frac{m}{8}-1}$ \\
  $\mathbf{\{3\}}$ & $\emptyset$ & $\emptyset$ \\
  $\{4\}$ & $\emptyset$ & $\emptyset$ \\
  $\mathbf{\{5\}}$ & $C$ & $C$ \\
  $\{6\}$ & $\emptyset$ & $\emptyset$ \\
  $\mathbf{\{7\}}$ & $A\cup B$ & $A\cup B$ \\
  $\{8\}$ & $\emptyset$ & $\emptyset$ \\
  $\mathbf{\{9\}}$ & $\{c_i\}$ & $\{c_i\}$ \\
  $\{10\}$ & $\emptyset$ & $\emptyset$ \\
  $\{1,2\}$ & $A\cup B$ & $A\cup B$ \\
  $\mathbf{\{1,5\}}$ & $C\cup a_{i\to i+\frac{m}{2}-1}\cup b_{i-\frac{m}{8}\to i+\frac{m}{2}-1}$ & $C\cup a_{i-\frac{m}{8}\to i+\frac{m}{2}-1}\cup b_{i-\frac{m}{8}\to i+\frac{m}{2}-1}$ \\
  $\{1,7\}$ & $A\cup B$ & $A\cup B$ \\
  $\mathbf{\{1,9\}}$ & $\{c_i\}\cup a_{i\to i+\frac{m}{2}-1}\cup b_{i-\frac{m}{8}\to i+\frac{m}{2}-1}$ & $\{c_i\}\cup a_{i-\frac{m}{8}\to i+\frac{m}{2}-1}\cup b_{i-\frac{m}{8}\to i+\frac{m}{2}-1}$ \\
  $\mathbf{\{2,5\}}$ & $C\cup a_{i-\frac{m}{2}\to i-1}\cup b_{i-\frac{m}{2}\to i+\frac{m}{8}-1}$ &
  $C\cup a_{i-\frac{m}{2}\to i+\frac{m}{8}-1}\cup b_{i-\frac{m}{2}\to i+\frac{m}{8}-1}$ \\
  $\{2,7\}$ & $A\cup B$ & $A\cup B$ \\
  $\mathbf{\{2,9\}}$ & $\{c_i\}\cup a_{i-\frac{m}{2}\to i-1}\cup b_{i-\frac{m}{2}\to i+\frac{m}{8}-1}$ &
  $\{c_i\}\cup a_{i-\frac{m}{2}\to i+\frac{m}{8}-1}\cup b_{i-\frac{m}{2}\to i+\frac{m}{8}-1}$ \\
  $\mathbf{\{5,7\}}$ & $V$ & $V$ \\
  $\{5,9\}$ & $C$ & $C$ \\
  $\mathbf{\{7,9\}}$ & $\{c_i\}\cup A\cup B$ & $\{c_i\}\cup A\cup B$ \\
  $\{1,5,2\}$ & $V$ & $V$ \\
  $\{1,5,7\}$ & $V$ & $V$ \\
  $\{1,5,9\}$ & $C\cup a_{i\to i+\frac{m}{2}-1}\cup b_{i-\frac{m}{8}\to i+\frac{m}{2}-1}$ & $C\cup a_{i-\frac{m}{8}\to i+\frac{m}{2}-1}\cup b_{i-\frac{m}{8}\to i+\frac{m}{2}-1}$ \\
  $\{1,9,2\}$ & $\{c_i\}\cup A\cup B$ & $\{c_i\}\cup A\cup B$ \\
  $\{1,9,7\}$ & $\{c_i\}\cup A\cup B$ & $\{c_i\}\cup A\cup B$ \\
  $\{2,5,7\}$ & $V$ & $V$ \\
  $\{2,5,9\}$ & $C\cup a_{i-\frac{m}{2}\to i-1}\cup b_{i-\frac{m}{2}\to i+\frac{m}{8}-1}$ & $C\cup a_{i-\frac{m}{2}\to i+\frac{m}{8}-1}\cup b_{i-\frac{m}{2}\to i+\frac{m}{8}-1}$ \\
  $\{2,9,7\}$ & $\{c_i\}\cup A\cup B$ & $\{c_i\}\cup A\cup B$ \\
    $\{5,7,9\}$ & $V$ & $V$ \\
    \bottomrule
  \end{tabular}
\end{table}

\begin{rem}
  In our construction of the graphs $G$ and $G'$, we work with
  segments that overlap for $1/8$th of the number of nodes. The
  critical reader will remark that an overlap of a single node would
  already be sufficient. Our choice for working with a larger overlap
  is indeed largely aesthetic.  Moreover, our proof still works for an
  extension of SHACL where shapes of the form $|r\cap E|\geq n$ would
  be allowed. This extension allows us to write shapes like
  $|\textsf{colleague}\cap\textsf{friend}| \geq 5$, stating that the
  node has at least five colleagues that are also friends. Such an
  extension then would still not be able to express full disjointness.
\end{rem}

\subsection{Further non-definability results}
\label{sec:furth-non-defin}

In Theorem~\ref{bool}, we showed that equality is primitive in
$\lang(\disj,\closed)$, and similarly, that disjointness is primitive
in $\lang(\eq,\closed)$. Can we strengthen these results to
$\lang(\fulldisj,\closed)$ and $\lang(\fulleq,\closed)$, respectively?
This turns out to be indeed possible.

That equality remains primitive in $\lang(\fulldisj,\closed)$ already
follows from our given proof in Section~\ref{sec:equality}. In effect,
the attentive reader may have noticed that we already cover full
disjointness in that proof. In contrast, our proof of primitivity of
disjointness in $\lang(\eq,\closed)$ does not extend to full
equality. Nevertheless, we can reuse our proof of primitivity of
\emph{full} disjointness as follows. The graphs $G$ and $G'$ from
Proposition~\ref{prop:fulldisj} are indistinguishable in
$\lang(\fulleq,\disj,\closed)$. Let $H$ and $H'$ be the same graphs
but with all directed edges reversed (i.e., the graphs illustrated in
Figure~\ref{fig:disjgraph2}). Then the same proof shows that $H$ and
$H'$ are indistinguishable in $\lang(\fulleq, \closed)$. However,
since $G$ and $G'$ are distinguishable by the inclusion statement
$\exists p^-.\top \subseteq \neg\disj(p^-,q^-)$, also $H$ and $H'$
are distinguishable by the inclusion statement
$\exists p.\top \subseteq \neg\disj(p,q)$. Thus, the primitivity of
disjointness in $\lang(\fulleq,\closed)$ is established.

\section{Extension to stratified recursion} \label{secrec}

Until now, we could do without shape names.  We do need them,
however, for recursive shape schemas.  Such schemas allow shapes
to be defined using recursive rules, much as in Datalog and logic
programming.  The rules have a shape name in the head; in the
body they have a shape that can refer to the same or other shape
names.

\begin{exa}
  The following rule defines a shape, named $s$, recursively:
  \[ s \gets \const c \lor (\eq(p,q) \land \exists r.s). \] A node $x$
  will satisfy $s$ if there is a (possibly empty) path of $r$-edges from $x$
  to the constant $c$, so that all nodes along the path satisfy
  $\eq(p,q)$ (for two property names $p$ and $q$).
\end{exa}

\paragraph*{Rules and programs}
We need to make a few extensions to our formalism and the
semantics.
\begin{itemize}
\item
  We assume an infinite supply $S$ of \emph{shape names}.  Again for
  simplicity of notation only, we assume that $S$ is disjoint from
  $N$ and $P$.
\item
  The syntax of shapes is extended
  so that \emph{every shape name is a shape}.
\item
  A vocabulary $\Sigma$ is now a subset of $N \cup P \cup S$; an
  interpretation $I$ now additionally assigns a subset $\semi s$ of $\dom^I$ to
  every shape name $s$ in $\Sigma$.
\end{itemize}

Noting the obvious parallels with the field of logic programming,
we propose to use the following terminology from that field.  A
\emph{rule} is of the form $s \gets \phi$, where $s$ is a shape
name and $\phi$ is a shape.  A \emph{program} is a finite set of
rules.  The shape names appearing as heads of rules in a program
are called the \emph{intensional} shape names of that program.

The following definitions of the semantics of programs are similar to
definitions well-known for Datalog.  A program is \emph{semipositive}
if for every intensional shape name $s$, and every shape $\phi$ in the
body of some rule, $s$ occurs only positively in $\phi$.  Let $\prog$
be a semipositive program over vocabulary $\Sigma$, with set of
intensional shape names $D$.  An interpretation $J$ over
$\Sigma \cup D$ is called a \emph{model} of $\prog$ if for every rule
$s \gets \phi$ of $\prog$, the set $\semj\phi$ is a subset of
$\semj s$.  Given any interpretation $I$ over $\Sigma - D$, there
exists a unique minimal interpretation $J$ that expands $I$ to
$\Sigma \cup D$ such that $J$ is a model of $\prog$ (Indeed, $J$ is
the least fixpoint of the well-known immediate consequence operator,
which is a monotone operator since $\mathcal{P}$ is semipositive
\cite{ahv_book}). We call $J$ the result of applying $\prog$ to $I$,
and denote $J$ by $\prog(I)$.

Stratified programs are essentially sequences of semipositive
programs.  Formally, a program $\prog$ is called
\emph{stratified} if it can be partitioned into parts $\prog_1$,
\dots, $\prog_n$ called \emph{strata}, such that
\textbf{(i)}
the strata have pairwise disjoint sets
of intensional shape names; \textbf{(ii)} each stratum is semipositive;
and \textbf{(iii)} the strata are ordered in such a way that when a shape
name $s$ occurs in the body of a rule in some stratum, $s$ is
not intensional in any later stratum.

Let $\prog$ be a stratified program with $n$ strata $\prog_1$,
\dots, $\prog_n$  and let again $I$ be an interpretation over a
vocabulary without the intensional shape names.  We define
$\prog(I)$, the result of applying $\prog$ to $I$, to be the
interpretation $J_n$, where $J_0 := I$ and $J_{k+1} :=
\prog_{k+1}(J_k)$ for $0\leq k<n$.

\paragraph*{Stratified shape schemas}  We are now ready to define
a \emph{stratified shape schema} again as a set of inclusions, but now
paired with a stratified program.  Formally, it is
a pair $(\prog,\tbox)$, where:
\begin{itemize} \item $\prog$ is a program that is stratified, and
  where every shape name mentioned in the body of some rule is
  an intensional shape name in $\prog$. \item $\tbox$ is a
  finite set of inclusion statements $\phi_1 \subseteq \phi_2$,
  where $\phi_1$ and $\phi_2$ mention only shape names that are
  intensional in $\prog$. \end{itemize}

Now we define a graph $G$ to \emph{conform} to $(\prog,\tbox)$
if $\iexpr{\phi_1}{\prog(G)}$
is a subset of $\iexpr{\phi_2}{\prog(G)}$, for every
inclusion $\phi_1 \subseteq \phi_2$ in $\tbox$.

\begin{rem}
  The nonrecursive notion of shape schema, defined in
  Section~\ref{secdefs}, corresponds to the special case where
  $\prog$ is the empty program.
\end{rem}

\paragraph*{Extending Theorem~\ref{bool}}

Theorem~\ref{bool} extends to stratified shape schemas.
Indeed, consider a stratified shape schema
$(\prog,\tbox)$.  Shapes not mentioning any shape names are
referred to as \emph{elementary shapes}.  We observe that for
every intensional shape
name $s$ and every graph $H$, there exists an elementary shape
$\phi$ such that $\iexpr s{\prog(H)} = \iexpr \phi H$.  Furthermore,
$\phi$ uses the same constants, quantifiers, and path expressions
as $\prog$.  For semipositive programs, this is shown using a
fixpoint characterization of the minimal model; for stratified
programs, this argument can then be applied repeatedly.  The
crux, however, is that graphs $G$ and $G'$ of
Proposition~\ref{deprop} will have the same $\phi$.  Indeed, by
that Proposition, the fixpoints of the different strata
will be reached on $G$ and on $G'$
in the same stage.  We effectively obtain an extension of
Proposition~\ref{deprop}, which establishes the theorem
for features $X$ other than $\closed$.

Also for $X=\closed$, the reasoning, given after Lemma~\ref{lemclos},
extends in the same way to stratified
shape schemas, since the graphs $G$ and $G'$ used there again
yield exactly the same evaluation for all shapes that do not use
$\closed$.

\paragraph*{Extending Theorem~\ref{thm:generalizedschema}}
Also Theorem~\ref{thm:generalizedschema} extends to stratified shape
schemas.  Thereto, Lemma \ref{lemclos} needs to be reproven in the
presence of a stratified program $\prog$ defining the intensional
shape names.  The extended Lemma~\ref{lemclos} then states that
$\iexpr\phi{\prog(G)}= \iexpr\phi{\prog(G')}$.  The proof of
Theorem~\ref{thm:generalizedschema} then goes through unchanged.

\paragraph*{Extending Theorem~\ref{thm:fulleqdisj}} Also
Theorem~\ref{thm:fulleqdisj} extends to stratified shape schemas for
the same reasons given above for Theorem~\ref{bool}.

\section{Concluding remarks} \label{seconc}

An obvious open question is whether our results extend further to
nonstratified programs, depending on various semantics that have been
proposed for Datalog with negation, notably well-founded or stable
models \cite{ahv_book,truszcz_stabwf}.  One must then deal with
3-valued models and, for stable models, choose whether the TBox should
hold in every stable model (skeptical), or in at least one
(credulous).  For example, Andre\c sel et al.\ \cite{andresel} adopt a
credulous approach.  In the same vein, even for stratified programs,
one may consider \emph{maximal} models instead of minimal ones, as
suggested for ShEx \cite{shex}.  Unified approaches developed 
for logic programming semantics can be naturally applied to SHACL
\cite{iclp_recursiveshacl}.

Notably, Corman et al.\ \cite{corman} have already suggested that
disjointness is redundant in a setting of recursive shape schemas with
nonstratified negation. Their expression is not correct, however
\cite{cormanbug}.\footnote{Their approach is to postulate two shape
  names $s_1$ and $s_2$ that can be assigned arbitrary sets of nodes,
  as long as the two sets form a partition of the domain.  Then for
  one node $x$ to satisfy the shape $\disj(E,p)$, it is sufficient
  that $E(x)$ is a subset of $s_1$ and $p(x)$ of $s_2$.  This
  condition is not necessary, however, as other nodes may require
  different partitions.}


A general question surrounding SHACL, even standard nonrecursive
SHACL, is to understand better in which sense (if at all) this
language is actually better suited for expressing constraints on RDF
graphs than, say, SPARQL ASK queries
\cite{shaclsparql,owlic,gent-valid-reason}.  Certainly, the affinity
with description logics makes it easy to carve out cases where higher
reasoning tasks become decidable \cite{leinberger,shaclsatsouth}.  It
is also possible to show that nonrecursive SHACL is strictly weaker in
expressive power than SPARQL\@.  But does SHACL conformance checking
really have a lower computational complexity?  Can we think of novel
query processing strategies that apply to SHACL but not easily to
SPARQL?  Are SHACL expressions typically shorter, or perhaps longer,
than the equivalent SPARQL ASK expression?  How do the expression
complexities \cite{vardi_comp} compare?

\section*{Acknowledgement}

We thank the anonymous referees for helpful comments that improved the
presentation of this paper. This research was supported by the
Flanders AI Research programme.

\bibliographystyle{alphaurl}
\bibliography{journal}

\appendix

\section{Supplementary Proofs}
\label{sec:supplementary-proofs}

\subsection*{Proof of Lemma~\ref{lem:strings}}

We first state an auxiliary lemma:

\begin{lem}
  \label{lem:bininf}
  
  Let $V$ be a finite set of $n$ elements, and let
  $R\subseteq V\times V$ be a binary relation over $V$. We have
  $R^* = R^0 \cup R^1 \cup \dots \cup R^{n-1}$.
\end{lem}

\begin{proof}
  $R^*$ is defined as $R^0 \cup R^1 \cup \dots$ however, we will
  show that if $(a,b)\in R^m$, with $m\geq n$, then there exists a
  $k < m$ such that $(a,b)\in R^k$.

  We call a sequence of elements $x_1,\dots,x_h$ an $R$-path if
  $(x_l,x_{l+1})\in R$ for $1\leq l\leq h$.

  If $(a,b)\in R^m$, then there exists an $R$-path
  $x_1,\dots,x_{m+1}$ with $x_1=a$ and $x_{m+1}=b$. As there are
  only $n$ total elements, there exists $i,j$ with
  $1\leq i<j\leq m+1$ such that $x_i=x_j$. Therefore,
  $x_1,\dots,x_{i-1},x_j,\dots,x_{m+1}$ is also an $R$-path. We
  conclude that $(a,b)\in R^{m-(j-i)}$, as desired.
\end{proof}

\begin{proof}[Proof of Lemma~\ref{lem:strings}.]
  The proof is by induction on the structure of $E$. Clearly for the
  base case $E = p$, we have the set $U=\{p\}$ and similarly for
  $E=p^-$ we have $U=\{p^-\}$. When $E=\id$, clearly
  $U=\{\id\}$. Next, we consider the inductive cases. When
  $E=E_1\cup E_2$, we know by induction there exists a set of
  strings $U_1$ for $E_1$, and $U_2$ for $E_2$. We then have
  $U=U_1\cup U_2$. When $E=E_1\comp E_2$, we again know by induction
  there exists a set of strings $U_1$ for $E_1$, and $U_2$ for
  $E_2$. We have
  $U = \{s_1\comp s_2\mid s_1\in U_1 \;\text{and}\; s_2\in U_2\}$.
  Finally, when $E=E'^*$, we know by induction there exists a set of
  strings $U'$ for $E'$. Let $W$ be a set of strings, we define
  $W^1:= W$, and for a natural number $m>1$,
  $W^m := \{s_1\comp s_2\mid s_1\in W, s_2\in W^{m-1}\}$. We also
  use the shorthand notation $E^m$, with $m>0$ a natural number, to
  denote $m$ compositions of the path expression $E$. For example,
  $E^3$ is $E\comp E\comp E$. By definition,
  $\iexpr{E'^*}{G}=\iexpr{\id}{G}\cup\iexpr{E'}{G}
  \cup\iexpr{E'^2}{G}\cup\dots$. By Lemma~\ref{lem:bininf} we know
  that this is the same as
  $\iexpr{E'^*}{G}=\iexpr{\id}{G}\cup\iexpr{E'}{G}
  \cup\iexpr{E'^2}{G}\cup\dots\cup\iexpr{E'^{n-1}}{G}$ for graphs
  with at most $n$ nodes. It then follows that
  $U=\{\id\}\cup U'\cup U'^2\cup\dots\cup U'^{n-1}$.
\end{proof}

\subsection*{Proof of Lemma~\ref{lem:disjoint_gprime}}

\begin{proof}
For $i=1,2,3,4$, define the \emph{i-th blob of nodes} to be the set
$X_i=\{x_i^1,\dots,x_i^M\}$ (see Figure~\ref{figraphs}).  We also
use the notations $\nxt(1)=2$; $\nxt(2)=3$; $\nxt(3)=4$;
$\nxt(4)=1$; $\prev(4)=3$; $\prev(3)=2$; $\prev(2)=1$; $\prev(1)=4$.
Thus $\nxt(i)$ indicates the next blob in the cycle, and $\prev(i)$
the previous.

The proof is by induction on the structure of $E$. If $E$ is a
property name, $E$ is simple so the claim is trivial. If $E$ is of
the form $p^-$, the first claim is clear because
$\iexpr{r^-}{G'}\subseteq \iexpr{E}{G'}$, and we only need to
verify the second one.  That holds because for any $i$, if
$v\in X_i$, then $\iexpr{p^-}{G'}(v)\supseteq X_{\prev(i)}$ and
clearly $X_{\prev(i)} - \iexpr{r}{G'}(v)\neq\emptyset$.  We next
consider the inductive cases.

First, assume $E$ is of the form $E_1 \cup E_2$.  When at least
one of $E_1$ and $E_2$ is not simple, the two claims immediately
follow by induction, since
$\iexpr{E}{G'} \supseteq \iexpr{E_1}{G'}$ and
$\iexpr{E}{G'} \supseteq \iexpr{E_2}{G'}$.  If $E_1$ and $E_2$ are
simple, then $E$ is simple and the claim is trivial.

Next, assume $E$ is of the form $E_1^*$.  If $E_1$ is not simple,
the two claims follow immediately by induction, since
$\iexpr{E}{G'} \supseteq \iexpr{E_1}{G'}$.  If $E_1$ is simple,
the first claim clearly hold for $E$, so we only need to verify
the second claim.  That holds because, by the form of $E$, every
node $v$ is in $\iexpr{E}{G'}(v)$, but not in $\iexpr{r}{G'}(v)$,
as $G$ does not have any self-loops.

Finally, assume $E$ is of the form $E_1 \comp E_2$.  Note that if
$E_1$ or $E_2$ is simple, clearly claim one holds because
$\iexpr{r}{G'} \subseteq \iexpr{E}{G'}$.  The argument that
follows will therefore also apply when $E_1$ or $E_2$ is simple.
We will be careful not to apply the induction hypothesis for the
second statement to $E_1$ and $E_2$.

We distinguish two cases.

\begin{itemize}
\item If $\iexpr{r}{G'} \subseteq \iexpr{E_2}{G'}$, then we show
  that $\iexpr{r}{G'} \subseteq \iexpr{E}{G'}$.  Let $v \in X_i$.
  We verify the following two inclusions:
  \begin{itemize}

  \item $\iexpr{E}{G}(v) \supseteq X_i$. Let $u\in X_i$. If
    $u\neq v$, choose a third node $w\in X_i$. Since $X_i$ is a
    clique, $(v,w)\in \iexpr{E_1}{G}$ because the first claim
    holds for $E_1$. By $\iexpr{r}{G'} \subseteq \iexpr{E_2}{G'}$,
    we also have $(w,u)\in \iexpr{E_2}{G'}$, whence
    $u\in \iexpr{E}{G'}(v)$ as desired. If $u=v$, we similarly
    have $(v,w) \in \iexpr{E_1}{G'}$ and
    $(w,u) \in \iexpr{E_2}{G'}$ as desired.

  \item $\iexpr{E}{G}(v) \supseteq X_{\nxt(i)}$. Let
    $u\in X_{\nxt(i)}$ and choose $w\neq v \in X_i$. Because the
    first claim holds for $E_1$, we have
    $(v,w) \in \iexpr{E_1}{G}$. By
    $\iexpr{r}{G'} \subseteq \iexpr{E_2}{G'}$, we also have
    $(w,u)\in \iexpr{E_2}{G'}$, whence $u\in \iexpr{E}{G'}(v)$ as
    desired.
  \end{itemize}
  We conclude that
  $\iexpr{E}{G'}(v) \supseteq X_i \cup X_{\nxt(i)} \supseteq
  \iexpr{r}{G'}$ as desired.

\item If $\iexpr{r^-}{G'} \subseteq \iexpr{E_2}{G'}$, then we show
  that $\iexpr{r^-}{G'} \subseteq \iexpr{E}{G'}$. This is
  analogous to the previous case, now verifying that
  $\iexpr{E}{G}(v) \supseteq X_i \cup X_{\prev(i)}$.
\end{itemize}

In both cases, the second statement now follows for every node
$v$.  Indeed, $v\in X_i \subseteq \iexpr{E}{G'}(v)$ but
$v\notin \iexpr{r}{G'}(v)$.
\end{proof}

\end{document}